%% file: qac0-writeup.tex
\documentclass[11pt]{article}
\input{preamble}
\newif\ifcomments
\commentstrue

\usepackage[margin=1in]{geometry}
\usepackage{tikz}
\usepackage{verbatim}
\usepackage{cancel}
\usepackage{titling}
\usepackage{multirow}
\usepackage{wasysym}
\usepackage{cancel}
\usepackage{quantikz}
\usepackage{iftex}
\usepackage[textsize=tiny]{todonotes}

\renewcommand{\complement}{c}

\usetikzlibrary{decorations.pathreplacing}

\newcommand{\weight}[1]{\mathbf{W}^{#1}}

\newcommand{\choir}{\Phi}
\newcommand{\choirE}{\choir_\calE} 
\newcommand{\chois}{\rho}

\newcommand{\chan}[1]{\calE_{#1}} %

\newcommand{\bellstate}{\mathrm{EPR}}
\newcommand{\EPR}{\bellstate}

\newcommand{\qalg}{\textsc{Quantum-Low-Degree-Algorithm}\xspace}

\newcommand{\ignore}[1]{}
\definecolor{cb-salmon-pink}{RGB}{255, 182, 119}

\ifcomments
    \newcommand{\nat}[1]{{\bf \textcolor{violet}{Nat:}} {#1}}
    \newcommand{\snote}[1]{\footnote{{\bf \color{blue}Shivam}: {#1}}}
    \newcommand{\natf}[1]{\footnote{{\bf \color{violet}Nat:} {#1}}}
    \newcommand{\shivam}[1]{{{\bf \color{blue}Shivam}: {#1}}}
    \newcommand{\fran}[1]{\textcolor{red}{[#1]}}
    \newcommand{\hnote}[1]{{\bf \textcolor{magenta}{[HY: #1]}}}
\else
    \newcommand{\nat}[1]{}
    \newcommand{\snote}[1]{}
    \newcommand{\natf}[1]{}
    \newcommand{\shivam}[1]{}
    \newcommand{\fran}[1]{}
    \newcommand{\hnote}[1]{}
\fi

\newcommand{\choispace}{\mathfrak{C}_{\nin,\nout}}

\newcommand{\paulis}[1][n]{\mathcal{P}_{#1}}
\newcommand{\Tr}{\mbox{\rm Tr}}

\newcommand{\supp}{\mathrm{supp}}
\newcommand{\Parity}{\mathrm{Parity}}
\newcommand{\Majority}{\mathrm{Majority}}
\newcommand{\unitary}[1]{\calU_{#1}}
\newcommand{\inn}{\mathrm{in}}
\newcommand{\junk}{\mathrm{junk}}
\newcommand{\out}{\mathrm{out}}
\newcommand{\aux}{\mathrm{aux}}

\newcommand{\size}{s}

\renewcommand{\proj}[1]{\ket{#1}\!\!\bra{#1}}

\title{On the Pauli Spectrum of $\mathsf{QAC}^0$}
\author{
Shivam Nadimpalli\thanks{Columbia University. Email: \url{sn2855@columbia.edu}.} 
\and 
Natalie Parham\thanks{Columbia University. Email: \url{natalie@cs.columbia.edu}.}
\and 
Francisca Vasconcelos\thanks{UC Berkeley. Email: \url{francisca@berkeley.edu}.}
\and 
Henry Yuen\thanks{Columbia University. Email: \url{hyuen@cs.columbia.edu}.}
\vspace{1em}}

\begin{document}

\pagenumbering{gobble}
\hypersetup{linkcolor={black}}

\maketitle
\begin{abstract}
The circuit class $\mathsf{QAC}^0$ was introduced by Moore (1999) as a model for constant depth quantum circuits where the gate set includes many-qubit Toffoli gates. Proving lower bounds against such circuits is a longstanding challenge in quantum circuit complexity; in particular, showing that polynomial-size $\mathsf{QAC}^0$ cannot compute the parity function has remained an open question for over 20 years.

In this work, we identify a notion of the \emph{Pauli spectrum} of $\mathsf{QAC}^0$ circuits, which can be viewed as the quantum analogue of the Fourier spectrum of classical $\mathsf{AC}^0$ circuits. We conjecture that the Pauli spectrum of $\mathsf{QAC}^0$ circuits satisfies \emph{low-degree concentration}, in analogy to the famous Linial, Mansour, Nisan (LMN) theorem on the low-degree Fourier concentration of $\mathsf{AC}^0$ circuits. If true, this conjecture immediately implies that polynomial-size $\mathsf{QAC}^0$ circuits cannot compute parity.

We prove this conjecture for 
the class of depth-$d$, polynomial-size $\mathsf{QAC}^0$ circuits with at most $n^{O(1/d)}$ auxiliary qubits. We obtain new circuit lower bounds and learning results as applications:  this class of circuits cannot correctly compute
\begin{itemize}
    \item the $n$-bit parity function on more than $(\frac{1}{2} + 2^{-\Omega(n^{1/d})})$-fraction of inputs, and
    \item the $n$-bit majority function on more than $(1 - \Omega(n^{-1/2}))$-fraction of inputs.
\end{itemize}
Additionally we show that this class of $\mathsf{QAC}^0$ circuits with limited auxiliary qubits can be learned with quasipolynomial sample complexity, giving the first learning result for $\mathsf{QAC}^0$ circuits.

More broadly, our results add evidence that ``Pauli-analytic'' techniques can be a powerful tool in studying quantum circuits.

\end{abstract}

\hypersetup{linkcolor={purple}}

\newpage

\pagenumbering{arabic}

\input{sections/intro}

\section*{Acknowledgements}
We thank Rocco Servedio for sagacious feedback. We thank Joe Slote for discussions on connections between quantum circuits and Fourier analysis. We thank Daniel Grier, and Gregory Rosenthal for helpful conversations. N.P. thanks Sergey Bravyi and Chinmay Nirkhe for thoughtful discussions. F.V. thanks Hsin-Yuan Huang for informative discussions on learning. We thank anonymous reviewers for their helpful feedback. We thank Mauricio Soler for helpful feedback and discussion. We thank Chirag Wadhwa for pointing out an issue with our majority bound in a previous version of this paper, amongst other helpful comments. S.N.~is supported by NSF grants IIS-1838154, CCF-2106429, CCF-2211238, CCF-1763970, and CCF-2107187. F.V. is supported by NSF grant DGE-2146752 and the Vannevar Bush Faculty Fellowship Program grant N00014-21-1-2941. N.P. and H.Y. are supported by AFOSR award FA9550-21-1-0040, NSF CAREER award CCF-2144219, and the Sloan Foundation. This work was partially completed while the authors were visiting the Simons Institute for the Theory of Computing. 

\input{sections/prelims}

\input{sections/low-deg-conc}
\input{sections/parity-lb}

\input{sections/learning}

\bibliographystyle{alpha}
\bibliography{references}

\appendix
\input{sections/combination-channels}

\end{document}

%% file: preamble.tex
\usepackage{graphicx}
\usepackage[breaklinks=true, linktocpage=true, pagebackref]{hyperref}
\usepackage{amssymb}
\usepackage{tikz}
\usepackage{commath}
\usepackage{braket}
\usepackage{wasysym}
\usepackage{booktabs}
\usepackage{amsmath}
\usepackage{amsthm}
\usepackage{lipsum}
\usepackage{xspace}
\usepackage{url}
\usepackage[nameinlink]{cleveref}
\usepackage{bbm}

\hypersetup{
    colorlinks = true,
    linkcolor={purple},
    urlcolor={purple},
    citecolor={blue!80!black}
}

\newcommand{\wt}[1]{\widetilde{#1}}
\newcommand{\wh}[1]{\widehat{#1}}
\newcommand{\sse}{\subseteq}
\newcommand{\bits}{\{\pm1\}}
\newcommand{\zo}{\{0,1\}}
\newcommand{\bn}{\bits^n}
\newcommand{\isafunc}{:\bn\to\bits}
\newcommand{\isazofunc}{:\zo^n\to\zo}

\newcommand{\E}{\mathbb{E}} 

\newcommand{\NN}{\mathbb{N}}
\newcommand{\CC}{\mathbb{C}}

\theoremstyle{definition}
\newtheorem{definition}{Definition}
\newtheorem{theorem}[definition]{Theorem}
\newtheorem{lemma}[definition]{Lemma}
\newtheorem{claim}[definition]{Claim}

\newtheorem{proposition}[definition]{Proposition}
\newtheorem{fact}[definition]{Fact}
\newtheorem{corollary}[definition]{Corollary}
\newtheorem{conjecture}{Conjecture}

\newtheorem{remark}[definition]{Remark}

\newtheorem{notation}[definition]{Notation}

\crefname{claim}{claim}{claims}

\renewcommand{\abs}[1]{{\left| #1 \right|}}

\newcommand{\abra}[1]{{\left\langle #1 \right\rangle}}
\newcommand{\pbra}[1]{{\left( #1 \right)}}
\newcommand{\cbra}[1]{{\left\{ #1 \right\}}}
\newcommand{\sbra}[1]{{\left[ #1 \right]}}

\newcommand{\ceil}[1]{{\left\lceil #1 \right\rceil}}

\newcommand{\fnorm}[1]{\norm{#1}_F}
\newcommand{\tnorm}[1]
{\norm{#1}_{\mathsf{tr}}}

\DeclareMathOperator*{\tr}{Tr}

\newcommand{\ketbra}[2]{|{#1}\rangle \langle {#2} |}
\newcommand{\proj}[1]{\ketbra{#1}{#1}}

\renewcommand{\abs}[1]{{\left| #1 \right|}}

\newcommand{\ACclass}{\mathsf{AC}} 
\newcommand{\QAC}{\mathsf{QAC}}
\newcommand{\QNC}{\mathsf{QNC}}
\newcommand{\ACZ}{\ACclass^0}
\newcommand{\NC}{\mathsf{NC}}
\newcommand{\NCZ}{\NC^0}
\newcommand{\QNCZ}{\QNC^0}
\newcommand{\QACZ}{\QAC^0}

\newcommand{\CZ}{\textsf{CZ}\xspace}

\newcommand{\calM}{\mathcal{M}}
\newcommand{\calN}{\mathcal{N}}

\newcommand{\calP}{\mathcal{P}}

\newcommand{\calU}{\mathcal{U}}
\newcommand{\calE}{\mathcal{E}}
\newcommand{\calF}{\mathcal{F}}
\newcommand{\calB}{\mathcal{B}}

\newcommand{\C}{\mathbb{C}}

\newcommand{\one}{\mathbbm{1}}

\newcommand{\bx}{\boldsymbol{x}}

\renewcommand{\E}{\operatorname{{\bf E}}}
\renewcommand{\Pr}{\operatorname{{\bf Pr}}}

\newcommand{\Ex}{\mathop{{\bf E}\/}}
\newcommand{\Prx}{\mathop{{\bf Pr}\/}}

\newcommand{\poly}{\text{poly}}
\newcommand{\pauli}{\mathcal{P}}
\newcommand{\nout}{\ell}
\newcommand{\Nout}{L}
\newcommand{\nin}{n}
\newcommand{\Nin}{N}
\newcommand{\ntot}{m}
\newcommand{\Ntot}{M}

\newcommand{\diagram}[1]{{\begin{tikzpicture}[scale=.5,every node/.style={sloped,allow upside down},baseline={([yshift=+0ex]current bounding box.center)}] #1 \end{tikzpicture}}}

\definecolor{tensorblue}{rgb}{0.8,0.8,1}
\definecolor{tensorred}{rgb}{1,0.5,0.5}
\definecolor{tensorpurp}{rgb}{1,0.5,1}

\tikzset{ten/.style={fill=tensorblue}}
\tikzset{tenred/.style={fill=tensorred}}
\tikzset{tengreen/.style={fill=green!50!black!50}}
\tikzset{tenpurp/.style={fill=tensorpurp}}
\tikzset{tengrey/.style={fill=black!20}}
\tikzset{tenorange/.style={fill=orange!30}}
\tikzset{u/.style={fill=blue!20,draw=black}}
\tikzset{w/.style={fill=green!50!black!80,draw=black}}

%% file: sections/intro.tex
\section{Introduction}
\label{sec:intro}

The Fourier spectrum of a Boolean function provides highly useful information about its complexity. For example, the celebrated result of Linial, Mansour, and Nisan~\cite{LMN:93} shows that polynomial-size $\ACZ$ circuits give rise to functions whose Fourier spectra obey \emph{low-degree concentration}; in other words, they are approximately low-degree polynomials. Since then, Fourier analytic techniques have played an essential role in breakthroughs in 
learning theory~\cite{KushilevitzMansour:93,KKMS:05,eskenazis2022low}, pseudorandomness~\cite{Braverman:09,chattopadhyay2018pseudorandom,chattopadhyay2019pseudorandom}, property testing~\cite{FLNRRS,Blaisstoc09,KMS18}, classical-quantum separations~\cite{raz2022oracle}, and more. 

Are there analogous analytic techniques for studying models of quantum computation? A natural quantum generalization of the Fourier spectrum of a Boolean function is the \emph{Pauli spectrum} of a many-qubit operator. Recall that the set of $n$-qubit Pauli operators, $\calP_n := \cbra{I,X,Y,Z}^{\otimes n}$, forms an orthonormal basis for the space of $2^n\times 2^n$ complex matrices, with respect to the (normalized) Hilbert--Schmidt inner product. 
Consequently, any $n$-qubit operator $A$ can be decomposed as $A = \sum_{P\in\calP_n} \wh{A}(P) \, P$, analogous to the Fourier expansion of a Boolean function.
Our paper is motivated by the following question:
\begin{center}
    \emph{When does the Pauli spectrum reveal useful information \\ about the complexity of a quantum computation?}
\end{center}

Analyzing the Pauli spectrum of quantum operations has been a fruitful endeavor in recent years, leading to structural results about so-called quantum Boolean functions~\cite{montanaro2008quantum,rouze2022quantum}, learning algorithms for quantum dynamics~\cite{huang2022learning,volberg2023noncommutative,slote2023noncommutative}, property testing of quantum operations~\cite{wang2011property,chen2023testing,bao2023nearly}, and classical simulations of noisy quantum circuits~\cite{aharonov2023polynomial}. Although each result uses a slightly different notion of Pauli decomposition, most of them focus on the Pauli spectrum of unitary operators. However, it is unclear how this notion connects with the complexity of the unitary operator. 

\paragraph{An Illustrative Example.} One might have hoped for the following Linial-Mansour-Nisan-style low-degree concentration statement: if a unitary $U$ is computable by some simple circuit (for some notion of ``simple''), then most of its ``Pauli mass'' should concentrate on its low-degree part.\footnote{The \emph{degree} of a Pauli tensor $P\in\paulis$, denoted $|P|$, is the number of qubits on which it acts non-trivially (i.e. the number of non-identity components).} Unfortunately, however, such notions of low-degree Pauli concentration break down for even the simplest unitaries: consider the unitary $U = X^{\otimes n}$, which can be implemented by a single layer of single-qubit gates (see \Cref{fig:potato}). Clearly the Pauli mass of $U$ is concentrated on a single degree-$n$ coefficient, yet this unitary $U$ is computed by an extremely simple circuit. 

\begin{figure}[h]
    \centering
    \begin{quantikz}[row sep = 2mm]
		\qw &\qw & \qw & \qw & \gate X  & \qw & \qw & \qw \\
		\qw &\qw & \qw & \qw & \gate X  & \qw & \qw & \qw \\[-0.25em]
            &&&& \vdots &&&\\
		\qw &\qw & \qw & \qw & \gate X  & \qw & \qw & \qw 
    \end{quantikz}  
    \caption{A simple unitary that does not satisfy low-degree Pauli concentration in the naive sense.}
    \label{fig:potato}
\end{figure}
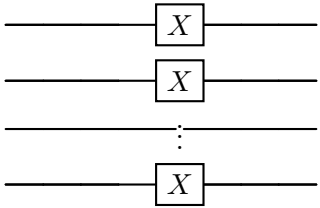

Note, however, that in this example there is an incongruity between the classical and quantum settings. A classical Boolean function $f\isazofunc$ only has one output bit, whereas the output of unitary $U$ is $n$ qubits. This leads us to restrict our focus to quantum circuits where there is a designated ``target'' qubit in the output. 
In the aforementioned example, while unitary $U$ has degree-$n$, the target qubit is not influenced by any other qubits, meaning $U$ should ``morally'' have degree-$1$. This suggests that instead of analyzing the unitary operator corresponding to the circuit, one should analyze the \emph{quantum channel} that comes from applying the circuit and then tracing out all qubits except for the target qubit.

\paragraph{A New Notion of Pauli Spectrum.} 
We introduce a different notion of Pauli decomposition: The Pauli spectrum of a quantum channel $\calE$ mapping $n$ qubits to $\ell$ qubits is the set of Pauli coefficients of its $(n+\ell)$-qubit \emph{Choi representation} $\Phi_{\calE} := (I^{\otimes n} \otimes \calE)(\ketbra{\mathrm{EPR}_n}{\mathrm{EPR}_n})$ which is the channel applied to half of unnormalized $n$-qubit Bell state $\ket{\EPR_n} := \sum_{x \in \{0,1\}^n} \ket{x} \otimes \ket{x}$ (see \Cref{subsec:choi-stuff,subsec:choi-pauli} for formal definitions of Choi matrices and their Pauli spectrum).

We show that this definition of Pauli spectrum connects with computational complexity much more closely. In particular we illustrate the definition's usefulness by studying the Pauli spectrum of $\QACZ$ circuits. These are a model of shallow quantum circuits that, in addition to single-qubit gates, can include wide Toffoli gates acting on arbitrarily many qubits~\cite{moore1999quantum}. This is a quantum analogue of the classical circuit model $\ACZ$, and represents the frontier of quantum circuit complexity: although we know that polynomial-size $\ACZ$ circuits cannot compute parity~\cite{FSS:84,Hastad:86}, proving the same for $\QACZ$ circuits remains a longstanding open problem~\cite{moore1999quantum,fang2003quantum,hoyer2005quantum,bera2007small,pade2020depth,rosenthal2020bounds}. 

\paragraph{Our Results, In a Nutshell.} Our main technical result, at a high level, is that polynomial-size, single-qubit-output $\QACZ$ circuits that use few auxiliary qubits must have Pauli spectrum that is highly concentrated on low-degree coefficients. This is a new structural result about $\QACZ$ circuits (with limited auxiliary qubits) that immediately yields average-case circuit lower bounds: we show that such circuits cannot compute the parity or majority functions, even approximately (see \Cref{subsec:our-results} for detailed theorem statements).

This raises the intriguing question of whether low-degree concentration holds for \emph{all} polynomial-size $\QACZ$ circuits, not just ones with few auxiliary qubits. We conjecture that this is indeed true (see \Cref{conj:qac0-spetral-conc} for a formal statement). This would be directly analogous to the Linial-Mansour-Nisan theorem about the low-degree concentration of $\ACZ$ circuits~\cite{LMN:93}, and would immediately show that parity is not computable by polynomial-size $\QACZ$ circuits, resolving the central open question posed in Moore's 1999 paper introducing the $\QAC$ circuit model in the first place~\cite{moore1999quantum}. 

Finally, we show that low-degree concentration of the Pauli spectrum of quantum channels yields sample-efficient \emph{learning algorithms} for those channels (see \Cref{thm:learning} for a more detailed theorem statement). This also directly corresponds to the learning result of~\cite{LMN:93} who show that low-degree concentrated Boolean functions can be learned with quasipolynomial complexity. Our results directly imply that $\QACZ$ circuits with few auxiliary qubits can be learned using quasipolynomial sample complexity, and if the conjectured low-degree concentration of $\QACZ$ holds, then \emph{all} polynomial-size $\QACZ$ circuits are sample-efficiently learnable. 

\vspace{10pt}

Although we weren't able to prove ``quantum LMN,'' we believe that the conjecture provides tantalizing motivation for studying the analytic structure of $\QACZ$ circuits, and for further investigating this notion of Pauli spectrum more broadly. The analogy with classical Fourier analysis of Boolean functions appears quite strong; it will be interesting to discover how far the analogy goes. 

Before explaining our results in more detail, we give a brief overview of $\QACZ$ circuits.

\subsection{\texorpdfstring{$\QACZ$}{QAC0} circuits}
\label{subsec:background}

The $\QACZ$ circuit model consists of constant-depth quantum circuits with arbitrary single-qubit gates and arbitrary-width Toffoli gates, which implement the unitary transformation 
\[\ket{x_1,\ldots,x_n,b} \mapsto \ket{x_1,\ldots,x_n,b \oplus \bigwedge_i x_i}.\] 
$\QACZ$ and related models were first introduced by Moore~\cite{moore1999quantum} to explore natural quantum analogues of classical circuit classes such as $\NCZ$, $\ACZ$, and $\ACZ[q]$, which are well-studied models of shallow circuits in classical complexity theory. 

Aside from being a natural generalization of a classical circuit model, $\QACZ$ also gives a clean theoretical model with which to study the power of many-qubit operations in quantum computations. Recently there has been increasing motivation for understanding the power of short-depth quantum computations with many-qubit or non-local operations. Some experimental platforms are beginning to realize controllable operations that can affect many qubits at once; examples include analog simulators~\cite{bluvstein2022quantum}, ion traps~\cite{gokhale2021quantum,guo2022implementing}, and superconducting qubit platforms that have mid-circuit measurements~\cite{rudinger2022characterizing}.

\paragraph{Parity versus $\QACZ$.}

Already in his 1999 paper~\cite{moore1999quantum}, Moore posed the question of whether the $n$-bit parity function can be computed in $\QACZ$. Recall that computing parity is equivalent to computing fanout (i.e. the unitary operation
$\ket{b,x_1,\ldots,x_n} \mapsto \ket{b,x_1 \oplus b,\ldots,x_n\oplus b}$) up to conjugation by Hadamard gates~\cite{moore1999quantum}. Consequently if $\QACZ$ circuits could compute parity, then this would imply that $\QACZ$ would be remarkably powerful: among other things, they would be capable of generating GHZ states, computing the $\mathsf{MOD}_q$ function for all $q$, computing the parity function, and performing phase estimation and approximate quantum Fourier transforms~\cite{moore1999quantum,bera2007small}.\looseness=-1

Despite being open for more than twenty years, this question has seen limited progress. Fang et al.~\cite{fang2003quantum} showed that $\QAC$ circuits with sublinear auxiliary qubits cannot compute parity, and more recently, Pad\'e et al.~\cite{pade2020depth} showed that depth-$2$ $\QAC$ circuits with arbitrary auxiliary qubits cannot compute parity. Rosenthal~\cite{rosenthal2020bounds} proved that parity can be approximated exponentially well by $\QACZ$ circuits that have exponentially many auxiliary qubits (it was not known before whether $\QACZ$ circuits of \emph{any} size could approximate parity). Rosenthal also proved that certain classes of $\QACZ$ circuits require exponential size to approximate parity; however, extending these lower bounds to general $\QACZ$ circuits seems challenging. (See \Cref{subsec:related-work} for more details about these prior works.)

Although lower bounds against classical $\ACZ$ are considered one of the great successes of complexity theory~\cite{Yao:85,Hastad:86,Raz87,Smol87}, it is far from clear how to extend those techniques (such as switching lemmas) to the setting of $\QACZ$. Furthermore, it is unclear whether $\QACZ$ lower bounds imply $\ACZ$ lower bounds: we do not know if $\QACZ$ circuits can even simulate classical $\ACZ$ circuits. Even though $\QACZ$ circuits appear quite weak (especially for classical computation), lower bounds have been difficult to obtain.

\paragraph{Going Beyond Lightcones.} 
On the other hand, if we restrict ourselves to constant-depth quantum circuits with only two-qubit gates (known as $\QNCZ$ circuits), it is comparatively much easier to prove limitations. For example, such circuits cannot prepare states with long-range entanglement (like the many-qubit GHZ state or states with topological order)~\cite{eldar2017local,verresen2112efficiently,anshu2023nlts}. At the heart of all $\QNCZ$ lower bound techniques, is the ``lightcone argument''--- the observation that each output qubit can only be affected by a small number of input qubits since there are only a few layers of small gates (see~\Cref{fig:lightcone}). This argument immediately breaks when either the circuit has logarithmic depth, or large many-qubit gates. Thus, any effort to prove lower bounds against more general circuits calls for novel techniques \emph{beyond lightcones}. $\QACZ$ circuits are at the frontier of this boundary and thus an attractive point of attack for developing new circuit lower bound techniques.

\begin{figure}
\centering
\resizebox{!}{3.3cm}{
\begin{tikzpicture}[scale=0.6]
\fill[fill=yellow, opacity=0.2] (5,1.425) -- (-5, 3.5) -- (-5, -0.5);
\draw[dashed] (5,1.425) -- (-5, 3.5);
\draw[dashed] (5,1.425) -- (-5, -0.5);

\node (as) at (0,0) {
\begin{quantikz}[row sep = 1.5em, column sep=2em]
    & \gate[2]{} & \qw & \qw & \qw & \qw & \gate[2]{} & \qw \\
    & \qw & \gate[2]{} & \qw & \qw & \qw & \qw & \qw \\
    & \qw & \qw & \qw & \gate[2]{} & \qw & \qw & \qw \\
    & \qw & \gate[2]{} & \qw & \qw & \qw & \gate[2]{} & \qw \\
    & \qw & \qw & \qw & \qw & \qw & \qw & \qw 
\end{quantikz}            
};
\end{tikzpicture}
}
\caption{A (backwards) lightcone of a qubit.}\label{fig:lightcone}
\end{figure}
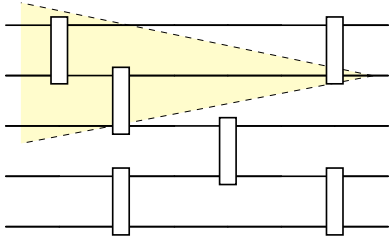

\subsection{Our Results}
\label{subsec:our-results}

\begin{table}[t]
\centering
\setlength{\tabcolsep}{15pt}
\renewcommand{\arraystretch}{1.5}
\begin{tabular}{@{}lll@{}}
\toprule
& \begin{tabular}[t]{@{}l@{}}$\ACZ$ Circuits\\[-0.35em] (Boolean function $f$)\end{tabular}          & \begin{tabular}[t]{@{}l@{}}$\QACZ$ Circuits\\[-0.35em] (Choi representation $\choir$)\end{tabular}   \\ 
\midrule

Fourier Basis           
& \begin{tabular}[t]{@{}l@{}}$\chi_S (x) \in \cbra{\prod_{i\in S} x_i}_{S\sse[n]}$\\[-0.35em] (Parity functions) \end{tabular} 
& \begin{tabular}[t]{@{}l@{}}$P\in\{I,X,Y,Z\}^{\otimes n}$\\[-0.35em] (Pauli tensors) \end{tabular} \\[2em]
Decomposition           
& \begin{tabular}[t]{@{}l@{}}$f(x)=\underset{S\sse[n]}{\sum} \widehat{f}(S)\chi_S (x) $\\[-0.35em] (Fourier decomposition) \end{tabular} 
& \begin{tabular}[t]{@{}l@{}}$\Phi=\underset{P\in\paulis}{\sum} \widehat{\Phi}(P)\cdot P $\\[-0.35em] (Pauli decomposition) \end{tabular} \\[2em]
Spectral Concentration  
&  \begin{tabular}[t]{@{}l@{}}$\weight{>k}[f] \leq \size\cdot 2^{-\Omega\pbra{k^{1/d}}}$\\[-0.35em] (Lemma 7 of \cite{LMN:93}) \end{tabular}                        
&  \begin{tabular}[t]{@{}l@{}}$\weight{>k}[\choir] \leq \size^2\cdot 2^{-\Omega(k^{1/d} -a)}$\\[-0.35em] (\Cref{thm:informal-concentration}) \end{tabular} \\[2em]
Correlation with Parity 
&  \begin{tabular}[t]{@{}l@{}}$\frac{1}{2} + \size\cdot 2^{-\Theta\pbra{n^{1/d}}}$\\[-0.35em] (Implicit in \cite{LMN:93}) \end{tabular}                      
&  \begin{tabular}[t]{@{}l@{}}$\frac{1}{2} + \size \cdot 2^{-\Omega(n^{1/d} - a)}$\\[-0.35em] (\Cref{thm:parity-lb}) \end{tabular} \\[2em]
Learnability            
&  \begin{tabular}[t]{@{}l@{}}quasipolynomial samples\\[-0.35em] (Section 4 of \cite{LMN:93}) \end{tabular}                        
&  \begin{tabular}[t]{@{}l@{}}quasipolynomial samples\\[-0.35em] (\Cref{thm:learning}) \end{tabular} \\[2em]
\bottomrule
\end{tabular}
\caption{Comparison with the work of Linial, Mansour, and Nisan~\cite{LMN:93}; see~\Cref{subsec:technical-overview} for a detailed discussion. In both the classical and the quantum settings, we denote the number of inputs as $n$, number of gates as $\size$, and the depth of the circuit as $d$. 
In the $\QACZ$ learning setting, we assume that $a\leq O(\mathrm{polylog}(n))$.} 
\label{tab:dictionary}
\end{table}

We show that the spectral properties -- with our notion of Pauli spectrum -- of $\QACZ$ circuits with limited auxiliary qubits resemble those of classical $\ACZ$ circuits; this in turn allows us to derive circuit lower bound and learning conclusions that are analogous to those of~\cite{LMN:93} (see~\Cref{tab:dictionary} for a comparison).  
In particular, our main technical result is the following bound on the Pauli coefficients of a $\QACZ$ circuit. For formal definitions of Choi representations, Pauli coefficients, etc, we refer the reader to \Cref{sec:prelims}.

\begin{theorem}[Informal version of \Cref{thm:pauli-concentration}] \label{thm:informal-concentration}
    Suppose a $n$ to $1$-qubit channel $\chan{}$ is computed by a depth-$d$ $\QACZ$ circuit with $a$ auxiliary qubits. Writing $\choir_{\chan{}}$ for the \emph{Choi representation} of $\chan{}$, we have 
    \[\sum_{\substack{P\in \paulis[n+1]} : \\|P| >k} \abs{\wh{\choir_{\chan{}}}(P)}^2 \leq 2^{-\Omega\pbra{k^{1/d}-a}} \]
    where $\{\wh{\choir_{\chan{}}}(P)\}$ is the collection of \emph{Pauli coefficients} of $\choir_{\chan{}}$. 
\end{theorem}

\Cref{thm:informal-concentration} can be interpreted as saying that most of the ``mass'' of the Pauli decomposition of the Choi representation of $\chan{}$ lies on Pauli tensors with few non-identity components. 
As a consequence of \Cref{thm:informal-concentration}, we obtain circuit lower bounds and learning algorithms for $\QACZ$ circuits.

\begin{theorem}[Informal version of \Cref{QAC-corr-bound}] \label{thm:parity-lb}
    Suppose $Q$ is a $\QAC$ circuit with depth $d = O(1)$ and at most $\frac{1}{2}\cdot  n^{1/d}$ auxiliary qubits. Let $Q(x)\in \zo$ denote the random measurement outcome in the computational basis of a single output qubit of $Q$ on input $\ket{x}$.

    \begin{itemize}
        \item $Q$ cannot {approximate} the $n$-bit parity function $\Parity_n(x) := \sum_{i=1}^n x_i \mod 2$, i.e.
        \[\Prx_{\bx\sim\zo^n}[Q(\bx) = \Parity_n(\bx)] \leq \frac{1}{2} +  2^{-\Omega\pbra{n^{1/d}}}.\]
        \item $Q$ cannot approximate the $n$-bit majority function  $\mathrm{\Majority}_n(x) := \mathbf{1}\cbra{\sum_{i=1}^n x_i \geq n/2}$, i.e.
        \[\Prx_{\bx\sim\zo^n}\sbra{Q(\bx) = \Majority_n(\bx)} \leq 1 - \Omega\pbra{n^{-1/2}}.\]
    \end{itemize}
\end{theorem}

We point out that our $\QACZ$ lower bounds in \Cref{thm:parity-lb} are \emph{average-case}: the circuits fail to approximate parity or majority. The only previously known average-case lower bounds for parity were shown by Rosenthal \cite{rosenthal2020bounds}. Notably, he showed an average-case bound against depth-2 $\QAC$ circuits and a size lower bound of $\Omega(n/d)$ for depth $d$ circuits. For a more detailed comparison between our lower bounds and previously established $\QACZ$ lower bounds, see \Cref{subsec:related-work}.

\vspace{5pt}

As a consequence of \Cref{thm:informal-concentration}, we also obtain an algorithm for learning $\QACZ$ circuits:

\begin{theorem}[Informal version of \Cref{thm:learnqac0}]
\label{thm:learning}
   Let $\calE$ be a channel computed by a depth-$d$ $\QACZ$ circuit on $n$ input qubits with $\mathrm{polylog}(n)$ auxiliary qubits. Then for $\delta>0$ and $\epsilon>\frac{1}{\text{poly}(n)}$, it is possible to learn a channel $\wt{\calE}$ satisfying 
   \begin{align}
       \sum_{P\in \paulis[n+1]} \abs{\wh{\Phi_{\calE}}(P)-\wh{\Phi_{\wt{\calE}}}(P)}^2 \leq \epsilon,
   \end{align} 
with probability $1-\delta$ using $n^{\mathrm{polylog} \left( n \right)}\log\left(1/\delta\right)$ copies of the Choi state $\frac{1}{N}\Phi_{\calE}$. Here, $\choirE$ and $\choir_{\wt{\calE}}$ are the Choi representations of $\calE$ and $\wt{\calE}$ respectively.
In the special case where $\calE$ computes a Boolean function $f: \zo^\nin \to \zo$, the learned channel $\wt{\calE}$ corresponds to a probabilistic function $g$ such that $\Pr_x[f(x) \neq g(x)] \leq O(\sqrt{\epsilon})$.
\end{theorem}

We further show, in \Cref{appendix:lin-comb-channels}, that all of the above results extend to quantum channels $\calE$ that are convex combinations of the channels $\{\calE_i\}$ implemented by $\QACZ$ circuits: $\calE(\rho) = \sum_i \alpha_i \calE_i(\rho)$. Note that it is not necessarily true that $\calE$ can be implemented by a $\QACZ$ circuit.
\subsection{Technical Overview}
\label{subsec:technical-overview}

The starting point for our results is the seminal work of Linial, Mansour, and Nisan~\cite{LMN:93} on the Fourier spectrum of constant-depth classical circuits; we refer the interested reader to the monographs~\cite{ODonnell2014,garban2014noise} for further background on the subject.

\subsubsection{The Work of Linial, Mansour, and Nisan (LMN)} \label{sec:LMN}

Suppose $f\isazofunc$ is a Boolean function computed by a depth-$d$ $\ACZ$ circuit of size $\size$. Recall that $f$---viewed as a function $f\isafunc$---can be expressed as a real multilinear polynomial 
\[f = \sum_{S\sse[n]} \wh{f}(S)\chi_S \qquad\text{where}\qquad \chi_S(x) := \prod_{i\in S} x_i\]
which can be viewed as a Fourier expansion of $f$. The main technical result of \cite{LMN:93}, namely Lemma~7, is the following bound on the Fourier spectrum of $f$:
\begin{equation} \label{eq:lmn}
    \sum_{|S| > k} \wh{f}(S)^2 \leq {\size\cdot 2^{-\Theta\pbra{k^{1/d}}}} \qquad\text{for all}~k\in[n],
\end{equation}
which, intuitively, states that most ``Fourier mass'' of $f$ lies on its low-degree part. Although this bound has been subsequently sharpened~\cite{boppana1997average,haastad2001slight,impagliazzo2012satisfiability,haastad2014correlation,tal2017tight}, for the sake of simplicity, this work will focus solely on the \cite{LMN:93} bound given by \Cref{eq:lmn}. 

The primary technical ingredient used by \cite{LMN:93} to prove the above bound is H{\aa}stad's celebrated switching lemma~\cite{Hastad:86}. 
As an immediate consequence of~\Cref{eq:lmn}, one can obtain correlation bounds against parity as well as a learning algorithm (under the uniform distribution) for $\ACZ$ circuits; we sketch both of these results below, see Section~4.5 of~\cite{ODonnell2014} for a thorough exposition.

\paragraph{Correlation Bounds Against Parity.} 

Suppose $f\isafunc$ is computed by a depth-$d$ size-$\size$ $\ACZ$ circuit. Recall that $\Parity_n\isafunc$ is given by $\Parity_n := \chi_{[n]}$. As an immediate consequence of \Cref{eq:lmn}, 
\[\abs{\wh{f}\pbra{[n]}} \leq \size\cdot 2^{-\Theta\pbra{n^{1/d}}}.\]
Furthermore, using Proposition~1.9 from~\cite{ODonnell2014}, is straightforward to check that 
\[\Prx_{\bx\sim\bn}\sbra{f(\bx) = \Parity_n(\bx)} \leq \frac{1}{2} + \size\cdot 2^{-\Theta\pbra{n^{1/d}}}.\]   
So, if $d = O(1)$ and $\size = \mathrm{poly}(n)$, then $f$ can agree with the parity function on at most $\frac{1}{2} + \exp(-\Theta(n))$ fraction of inputs. Since guessing $\zo$ uniformly at random gives correlation $1/2$ with the parity function, this result implies that, with constant-depth circuits, one cannot do much better than random guessing.

\paragraph{Learning $\ACZ$ Circuits.}

Taking $k = \Theta\pbra{\log^d\pbra{\frac{\size}{\epsilon}}}$, \Cref{eq:lmn} implies that $
\sum_{|S| > k} \wh{f}(S)^2 \leq \epsilon$.  In other words, all but $\epsilon$ of $f$'s ``Fourier weight'' lies on its low-degree coefficients. Based on this observation, \cite{LMN:93} suggest a natural learning algorithm for $\ACZ$ : Simply estimate all the low-degree Fourier coefficients $\{\wh{f}(S)\}_{|S|\leq k}$ to sufficiently high accuracy, and---writing $\widetilde{f}(S)$ for the estimate of $\wh{f}(S)$---output the $\bits$-valued function
\begin{align} \label{eqn:sng_function}
    \mathrm{sign}\pbra{\sum_{|S|\leq k}\widetilde{f}(S)\chi_S}.
\end{align}
This gives a quasipolynomial time algorithm for learning $\ACZ$ circuits. In fact, it is known that, under a strong enough cryptographic assumption, quasipolynomial time is \emph{required} to learn $\ACZ$ circuits even with query access~\cite{kharitonov1993cryptographic}.\footnote{Note that the \cite{LMN:93} learning algorithm only requires sample access to the function $f$, which is weaker than the class of algorithms with query access considered by Kharitonov~\cite{kharitonov1993cryptographic}.} 

\subsubsection{Spectral Concentration of \texorpdfstring{$\QACZ$}{QAC0} Circuits}
\label{sec:overview-spectral-concentration}

Inspired by the classical importance of low-degree Fourier concentration, i.e. \Cref{eq:lmn}, one might hope for an analogous notion of low-degree Pauli concentration in the quantum setting. As we saw earlier in the introduction via the example $U = X^{\otimes n}$, unitaries implemented by $\QACZ$ circuits do \emph{not} have Pauli weight that is low-degree concentrated. Instead, we turn to analyzing the Pauli spectrum of $\QACZ$ \emph{channels}.

\paragraph{The Pauli Decomposition of Channels.} 
We define the \emph{Pauli spectrum of channel $\calE$} as that of its Choi representation $\choirE$:\footnote{Recall that $\Phi_{\calE} = (I^{\otimes n} \otimes \calE)(\ketbra{\mathrm{EPR}_n}{\mathrm{EPR}_n})$ where $\ket{\EPR_n} = \sum_{x \in \{0,1\}^n} \ket{x} \otimes \ket{x}$ (see \Cref{subsec:choi-stuff})}
\begin{align*}
    \choirE = \sum_{P}\wh{\choirE}(P)P.
\end{align*}

To highlight that the Pauli spectrum of channels generalizes the classical Fourier spectrum of Boolean functions, 
note that when the channel computes a Boolean function $f :\zo^\nin \to \zo$, the Pauli spectrum of the Choi representation is closely related to the Fourier expansion of~$f$:
\begin{align}\label{eq:fourier-to-choi-intro}
    \choir_f = \frac{1}{2} I^{\otimes (\nin+1) } + \frac{1}{2}\sum_{S\subseteq [\nin]} \wh{f}(S) Z_S \otimes Z~.
\end{align}
Here, $Z_S$ denotes $\otimes_{i =1}^{\nin} Z^{\one\{i \in S\}}$. 

To further motivate our notion of Pauli spectrum, we return to our example $U = X^{\otimes n}$. Consider the channel $\chan{U}(\rho) = \tr_{[n-1]} (U \rho U^\dagger)$ that applies $U$ to the input state and traces out all but the last qubit. 
As a check, it is readily verified that
\[\Phi_{U} = I^{\otimes \nin-1} \otimes \sum_{y, y'\in\zo} \ketbra{y}{y'}\otimes Z\ketbra{y}{y'}Z.\]
Thus $\chan{U}$ is ``low-degree'' as originally hoped for.

\paragraph{The Pauli Spectrum of \texorpdfstring{$\QACZ$}{QAC0} Channels.} 
Returning to quantum circuits, suppose $\chan{}$ is implemented by a depth-$d$ $\QACZ$ circuit acting on $\nin$ input qubits and $a$ auxiliary qubits. Writing $\size$ for the number of Toffoli gates acting in the circuit, this work proves a bound on the Pauli spectrum of $\choir_{\chan{}}$. Namely, for each $k\in[\nin+1]$, we show that
\begin{equation} \label{eq:our-lmn1}
     \sum_{|P|>k} |\choir_{\chan{}}(P)|^2 
     \leq O\pbra{\size^2 2^{-k^{1/d} +a}}.
\end{equation}
Note the resemblance between \Cref{eq:our-lmn1} and \Cref{eq:lmn} obtained by \cite{LMN:93}. 
We now describe the proof of \Cref{eq:our-lmn1}. 

For simplicity, we first describe the case when $U$ is implemented by a circuit \emph{without} any auxiliary qubits, i.e. when $a = 0$. In this case, the proof proceeds in two steps:
\begin{enumerate}
    \item We first establish that if a depth-$d$ $\QACZ$ circuit does not have any Toffoli gates of width at least $k^{1/d}$, then it has no Pauli weight above level $k$; and 
    \item Next, we show that deleting such ``wide'' Toffoli gates does not noticeably affect the action of the circuit.
\end{enumerate}
A lengthy, albeit ultimately straightforward, calculation using standard analytic tools then establishes~\Cref{eq:our-lmn1} for the case when $a = 0$. We note that the proof of the first item above relies on a lightcone-type argument.

For the more general case where the circuit implementing $U$ has $a$ auxiliary qubits, we consider two cases, corresponding to \emph{clean} auxiliary qubits (i.e. when the $a$ qubits must start in the $\ket{0^a}$ state) and \emph{dirty} auxiliary qubits (i.e. when there is no guarantee on the setting of the state of the $a$ qubits, but regardless the circuit has the desired behavior).
\footnote{Perhaps surprisingly, dirty auxiliary states are a resource that can allow one to reduce circuit depth~\cite{morimae2014hardness,baked-potato}.} 
With clean auxiliary qubits, we can view the $a$ auxiliary qubits as a part of the input to the circuit that we enforce to be set to $\ket{0^a}$ by postselecting the Choi representation; this results in the $2^a$ blow-up in \Cref{eq:our-lmn1}.
In the dirty setting, however, we are able to incur no blowup; we defer discussion of this to \Cref{subsec:postselecting}.

\paragraph{New Circuit Lower Bounds.} As an immediate consequence of our spectral concentration bound on $\QACZ$, we obtain correlation bounds against parity and majority functions. This follows~by:
\begin{enumerate}
    \item First, relating the classical Fourier expansion of these functions to the Pauli expansion of the Choi representations of channels implementing those functions (as in \Cref{eq:fourier-to-choi-intro}); and
    \item Then, showing that they cannot be approximated by $\QACZ$ with a small number of auxiliary qubits thanks to the spectral concentration as discussed above. 
\end{enumerate}

\subsubsection{Learning \texorpdfstring{$\QACZ$}{QAC0} Circuits}
\label{subsubsec:learning-overview}

Our main learning result is an algorithm for learning channels with Pauli weight that is low-degree concentrated. 
Combining this with our low-degree concentration bound for $\QACZ$ channels (\Cref{eq:our-lmn1})  immediately provides a learning algorithm for $\QACZ$ channels (with limited auxiliary qubits). 
Specifically, we show that quantum channels from $\nin$ to $1$ qubits that are implemented with a $\QACZ$ circuit and $o(\mathrm{poly}\log(n))$ auxiliary qubits can be learned\footnote{By ``learning a channel'' we mean learning an approximation (in Frobenius distance) of its Choi representation.} using quasipolynomial, i.e. $2^{O(\mathrm{poly}\log \nin)}$, samples 
. In comparison, naive tomography would require $2^{O(n)}$ samples. 

\paragraph{A Quantum Low-Degree Algorithm.} Our algorithm is inspired by---and closely follows the structure of---the classical low-degree algorithm introduced by \cite{LMN:93}. 

The first step of the learning algorithm is to estimate all low-degree Pauli coefficients (i.e. all $P\in\calP_n$ for $|P|\leq \mathrm{poly}\log(n)$) to sufficiently high accuracy. We consider two different learning models:
\begin{itemize}
    \item \textbf{Choi state samples:} In the first model we are given copies of the Choi state $\chois_\calE = \frac{\choirE}{\tr(\choirE)}$. We apply Classical Shadow Tomography~\cite{Huang2020} to learn the Pauli coefficients.
    \item \textbf{Measurement queries:} In this model, the learner is allowed to query an input state $\rho$ and observable $O$ and is given the measurement outcome of $\calE(\rho)$ with respect to $O$. For this setting, we perform direct tomography on the quantum channel $\calE$ for specially chosen input states $\rho$ and measurement observables $O$. This model was concurrently explored in the context of arbitary quantum processes, under the name ``Quantum Process Statistical Queries (QPSQ)'' model \cite{wadhwa2024learning}. 
\end{itemize}
While our approaches for both settings achieve similar sample complexity, the latter has the benefit of being more feasible (from an experimental implementation standpoint) than the former. Furthermore, the first approach can be thought of as analogous to learning from random labeled examples, and the second can be thought of as learning from query access to the function. 

\paragraph{``Rounding'' to CPTP maps.} After obtaining estimates $\wt{\choir}(P)$ for each of the Pauli coefficients $\wh{\choir}(P)$, the final step is to round  
$\widetilde{\Phi}_{\calE}=\sum_{P} \widetilde{\Phi}_{\calE}(P) P$
to a Choi representation of a valid quantum channel--- that is, a completely-positive trace-preserving (CPTP) map. 
Since the set of all CPTP Choi representations is a convex set, we show that there exists a projection onto this set that will only reduce the distance of our learned state to the true Choi representation. However, it remains open whether there exists an algorithm to implement an exact rounding procedure in quasipolynomial time.
Note that this final step is analogous to the classical low-degree algorithm's use of the sign function, as in \Cref{eqn:sng_function}.

\subsection{Related Work}
\label{subsec:related-work}

\begin{table}[t]
\centering
\setlength{\tabcolsep}{10pt}
\renewcommand{\arraystretch}{1.5}
\begin{tabular}{@{}lllll@{}}
\toprule
          & Hard Function & Depth ($d$)  &  Restrictions & Guarantee    \\ \midrule
    \cite{fang2003quantum}      & Parity$_n$ & $O(1)$      & $o(n)$ auxbits                & Worst Case   \\
    \cite{pade2020depth}   & Parity$_n$   & $2$       & None             & Worst Case   \\
    \cite{rosenthal2020bounds}   & Parity$_n$   & $2$       & None             & Average Case   \\
    \cite{rosenthal2020bounds}   & Parity$_n$   & $O(1)$       & $O(n/d)$ gates             & Average Case   \\
\Cref{thm:parity-lb} & Parity$_n$ & $O(1)$    & $\tfrac{1}{2} n^{1/d}$  auxbits          & Average Case \\ 
\Cref{thm:parity-lb} & Majority$_n$ & $O(1)$    & $\tfrac{1}{2} n^{1/d}$  auxbits          & Average Case\\ \bottomrule
\end{tabular}
\caption{Lower bounds against $\QAC$ circuits. Here, ``auxbits'' refers to auxiliary qubits.}
\label{tab:summary}
\end{table}

\paragraph{Previous work on Pauli analysis.}
The original proposal for ``quantum analysis of Boolean functions" came from~\cite{montanaro2008quantum}, who proposed the study of Pauli decompositions (i.e. ``Pauli analysis") of \textit{Hermitian unitaries}. The key intuition is that these unitaries possess $\pm 1$ eigenvalues, which can be interpreted as the outputs of a classical Boolean function. 
Recently, \cite{rouze2022quantum} extended seminal results from classical analysis of Boolean functions to this quantum setting. Furthermore, \cite{chen2023testing} extended these Pauli analysis techniques to the setting of quantum \textit{non-Hermitian unitaries} for applications to quantum junta testing and learning.  
Unfortunately, however, as illustrated earlier in this section, the Pauli spectrum of a unitary $U$ does not cleanly connect with the complexity of implementing $U$. In this work, we instead look to the Pauli decomposition of \emph{channels}. 

Recently,~\cite{bao2023nearly} also look to Pauli analysis of quantum channels for junta property testing, extending the work of ~\cite{chen2023testing} to unitary channels.  They, too, propose Fourier analysis of superoperators in the Choi representation, though they utilize a different Fourier basis than is used in this work. Bao and Yao decompose channels into their \emph{Kraus representation} with Pauli Kraus operators. While this is a natural choice of basis for junta channels, it is not compatible with analyzing more general channels, or capturing circuit complexity. Notably, their analysis is limited to $n$ to $n$-qubit channels. In this work, we also consider Pauli analysis of superoperators. However, we instead focus solely on the \textit{Choi representation}, defining our Fourier spectrum in terms of the linear expansion of the Choi matrix itself into Pauli operators. By defining our Fourier basis in this way, we are able to analyze channels with an arbitrary number of output qubits. As discussed earlier in the section, this provides a definition of Pauli spectrum that connects more closely with computational complexity and allows us to prove spectral concentration, average-case lower bounds, and learning results for single-output-qubit $\QACZ$ circuits.


\paragraph{Previous work on $\QACZ$ lower bounds.} Since Moore's paper~\cite{moore1999quantum} that originally defined the model, there has only been a smattering of lower bound results on $\QAC$. We summarize known lower bounds against $\QACZ$ circuits below and in \Cref{tab:summary}. We then compare them with our lower bound results.

Fang, et al.~\cite{fang2003quantum} established the first lower bounds on the $\QACZ$ model; in particular they proved that a depth-$d$ $\QAC$ circuit cleanly computes the $n$-bit parity function with $a$ auxiliary qubits, then $d \geq 2\log (n/(a+1))$. Here, ``cleanly'' means that the auxiliary qubits have to start and end in the zero state. 

The key to their lower bound proof is a beautiful lemma (Lemma 4.2 of~\cite{fang2003quantum}): for all depth-$d$ $\QAC$ circuits, there exists a subset $S$ of $(a+1) 2^{d/2}$ input qubits and a state $\ket{\psi_S}$ for that subset $S$, such that no matter what the other input qubits are set to, the output and auxiliary qubits result in the zero state. This immediately implies a lower bound for $\QAC$ circuits that cleanly compute the parity function: First, the clean computation property implies that without loss of generality the subset $S$ is supported on non-auxiliary qubits. Second, if $d < 2 \log (n/(a+1))$, then there exists a non-auxiliary input qubit $i$ that is not fixed by $\ket{\psi_S}$, but the output qubit should depend on the state of the $i$'th qubit -- except the output is already fixed to zero, a contradiction. 

This lower bound is nontrivial as long as the number of auxiliary qubits is sublinear (i.e. $a = o(n)$), whereas our lower bound on the parity function can only handle up to $\sim n^{1/d}$ auxiliary qubits. On the other hand, the lower bound of~\cite{fang2003quantum} appears to be tailored to the setting where the circuit has to compute parity both exactly and cleanly. For a circuit that computes parity exactly (i.e. on all input strings), the clean computation property is without loss of generality because one can always save the output and then uncompute. When the circuit only computes parity approximately (e.g. on $\frac{1}{2} + \epsilon$ fraction of inputs), the clean computation property becomes an additional assumption. 

Furthermore, the technique of~\cite{fang2003quantum} does not obviously extend to obtain average-case lower bounds: although there may be a fixing $\ket{\psi_S}$ of $(a+1)2^{d/2}$ input qubits that force the output qubit to be zero, such a fixing occurs with probability at most $2^{-(a+1)2^{d/2}}$ under the uniform distribution on the $n$ input qubits -- note that this is \emph{exponentially} small in $a$ and \emph{doubly-exponentially} small in $d$. This directly implies that a depth-$d$ $\QAC$ circuit with $a$ auxiliary qubits cannot compute more than $1 - 2^{-(a+1)2^{d/2}}$ fraction of inputs. When $a=\omega(\log n)$ this fraction is extremely close to $1$. By comparison our average case lower bound shows that $\QAC$ circuits with limited auxiliary qubits cannot compute parity on more than $\frac{1}{2} + 2^{-\Omega(n^{1/d})}$ fraction of inputs. 

We note\footnote{We thank an anonymous reviewer for pointing this out to us.} that the techniques of~\cite{fang2003quantum} also yield a lower bound on cleanly computing the majority function, as fixing a sublinear number of input bits is not enough to fix the majority function. However, for the same reasons as mentioned in the previous paragraph, it is unclear whether this argument can be extended to prove an average-case lower bound. 

Later, Pad\'e et al.~\cite{pade2020depth} proved that no depth-$2$ quantum circuit (with \emph{any} number of auxiliary qubits) can cleanly compute parity in the worst case. They prove this by carefully analyzing the structure of states that can be computed by depth-$2$ $\QAC$ circuits. Similarly it is unclear whether these techniques can be extended to the non-clean or approximate computation setting.

Rosenthal~\cite{rosenthal2020bounds} proved that any average-case lower bound on $\QAC$ circuits (approximately) computing parity must use a bound on the number of auxiliary qubits; once there are \emph{exponentially many} auxiliary qubits, then there is a depth-$7$ $\QAC$ circuit approximately computing parity. Furthermore, Rosenthal proved the following average-case lower bounds:
\begin{enumerate}
    \item A depth-$d$ $\QAC$ circuit needs at least $\Omega(n/d)$ multiqubit Toffoli gates in order to achieve a $\frac{1}{2} + \exp(-o(n/d))$ approximation of parity, regardless of the number of auxiliary qubits. 

    \item Depth-$2$ $\QAC$ circuits, with any number of auxiliary qubits, cannot achieve $\frac{1}{2} + \exp(-o(n))$ approximation of parity, even non-cleanly. This proves an average-case version of the lower bound of~\cite{pade2020depth}. 

    \item A particular restricted subclass of $\QAC$ circuits (of which his depth-$7$ construction is an example) requires exponential size to compute parity, even approximately.    
\end{enumerate}
These are the first average-case lower bound results for $\QAC$ that we are aware of; however, they apply to restricted classes of $\QAC$ circuits and notably do not take into account the number of auxiliary qubits. As mentioned earlier, any general (average-case) lower bound on $\QAC$ circuits computing parity (for depths $7$ and greater) must depend on the number of auxiliary qubits. 

More recently, Slote \cite{slote2024parity} initiated the study of the closely related circuit class that is $\QNCZ$ circuits followed by $\ACZ$ post-processing, denoted $\ACZ \circ \QNCZ$. Slote conjectures that polynomial-sized $\ACZ \circ \QNCZ$ can not approximate parity, and shows that this is indeed the case when either the $\QNCZ$ circuit has no auxiliary qubits, or when the $\ACZ$ circuit has linear size. Perhaps surprisingly, the explicit connection between $\ACZ \circ \QNCZ$ and $\QACZ$ is unclear: while $\QACZ$ circuits can certainly implement $\QNCZ$ circuits, it is unknown whether they can implement $\ACZ$ --- it is, as far as we know, possible that $\QACZ$ is incomparable with both $\ACZ$ and $\ACZ \circ \QNCZ$. Nevertheless, for both $\QACZ$ and $\ACZ \circ \QNCZ$, many existing techniques (the lightcone argument) fail for similar reasons. Slote's approach utilizes Fourier analysis of Boolean functions and draws connections to nonlocal games.

\paragraph{Related work on quantum learning.} 
Efficient learning of quantum dynamics is a long-standing challenge in the field. 
Techniques such as quantum process tomography \cite{mohensi2008qpt}, which aim to fully characterize arbitrary quantum channels, require exponentially many data samples to guarantee a small error in the learned channel, for all possible channels. 

One way of achieving sample-efficient quantum channel learning algorithms is performing full tomography on \textit{specific classes} of quantum channels. By focusing on a specific class, rather than all possible quantum channels, there often exists nice structure which can be leveraged to reduce the number of data samples required to fully characterize channels in the class.
For example, \cite{bao2023nearly} showed that $n$-qubit to $n$-qubit quantum $k$-junta channels (acting non-trivially on at most $k$ out of $n$ qubits) can be learned to error $\epsilon$, with high probability, via $O(4^k/\epsilon^2)$ samples. In our learning result, we focus on quantum channels with an arbitrary number of output qubits, and with ``low-degree" Choi representations\footnote{We formally define the notion of a ``low-degree" Choi representation in \Cref{subsec:learning}.}. We show that a $k$-degree channel, involving $\ntot=\nin+\nout$ total input and output qubits, can be learned to error $\epsilon$, with high probability, via $\widetilde{O}((3\ntot)^k/(4^\nout \epsilon))$ samples. Our result is incomparable to \cite{bao2023nearly} since an $n$-qubit to $n$-qubit $k$-junta channel does not satisfy our notion of low-degree concentration. To our knowledge, this is the first work to analyze and offer a learning algorithm specific to channels with low-degree Choi representations. Furthermore, through our concentration result, we establish that $\QACZ$ circuits mapping to a single-qubit output ($\nout=1$) lie in this circuit class, resulting in the first quasipolynomial learning algorithm for single-qubit-output $\QACZ$.

An alternative approach for achieving sample-efficient learning algorithms is performing \textit{partial} tomography on arbitrary quantum channels. For example, rather than full process tomography of a channel $\calE$, \cite{huang2022learning} consider the task of learning the function $f(O_i,\rho)=\Tr(O_i\calE(\rho))$ for a class of $M$ observables $\{O_i\}_{i=1}^M$ and input state $\rho$. For an arbitrary quantum channel $\calE$ and bounded-degree observables of spectral norm $\|O_i\|\leq 1$, they prove that $2^{O(\log(1/\epsilon)\log(n))}$ samples are sufficient to learn the function for all observables to error $\epsilon$ with high probability. 
At a high level, our result and that of \cite{huang2022learning} both establish and leverage Fourier concentration (i.e.~low-degree approximation) to obtain efficient learning algorithms for quantum channels. 
However, our results operate in different settings. Namely,  our work learns a low-degree approximation of the channel's Choi representation, whereas theirs learns low-degree approximations of the channel's Heisenberg-evolved observables $O_i^*=\calE^\dag(O_i)$, where $f(O_i,\rho)=\Tr(O_i^*\rho)$. \cite{huang2022learning} show that under a locally-flat input distribution, the Heisenberg-evolved observables of general channels are well approximated by low-degree observables. While this enables efficient learning of \textit{any} quantum channel, restriction to locally-flat input distributions implies that, for quantum channels encoding classical Boolean functions, measurement expectations will be biased towards inputs which are not in the computational basis and, thus, uninformative. Our work instead obtains a sample-efficient learning result for the specific class of Choi representations of single-output-qubit $\QACZ$ circuits, with average-case guarantees according to the uniform distribution over computational basis states. 
To obtain this result, we prove low-degree concentration of $\QACZ$ Choi representations. This concentration result can further be shown to imply concentration of the channel's Heisenberg-evolved observables  
and, thus, could potentially be leveraged by the \cite{huang2022learning} procedure to offer a learning guarantee for single-qubit-output $\QACZ$ channels without restriction to locally-flat input distributions. It is an interesting direction of future research to formally relate the two works.

Finally, sample-efficient quantum channel 
learning algorithms can also be achieved by leveraging \textit{quantum-enhanced} experiments. \cite{chen22} proved exponential separations between learning algorithms with external quantum memory and those without. Building upon this, \cite{caro2023learning} recently demonstrated that full characterization of an unknown quantum channel's Pauli transfer matrix requires exponentially many channel queries in the case of classical processing and memories, but only polynomial samples in the case of quantum processing and memory. 
In this work, however, we do not consider quantum-enhanced experiments. Instead, we demonstrate that there exists a non-quantum-enhanced quasipolynomial learning algorithm for approximate characterization of the full Choi representation of single-output $\QACZ$ channels.

\subsection{Discussion}
\label{subsec:discussion}

We believe that much remains to be discovered about the analytic properties of $\QACZ$ circuits. We list some natural concrete (and not-so-concrete) questions for future work below:

\paragraph{Improved Spectral Concentration?}

Arguably the most natural open question is to improve the dependence on the number of auxiliary qubits in \Cref{thm:informal-concentration} so as to get a lower bound against $\QACZ$ circuits with polynomially many auxiliary qubits. In fact, we conjecture the following improved spectral bound:

\begin{conjecture}[Spectral concentration for $\QACZ$]\label{conj:qac0-spetral-conc}
    Suppose $\chan{}$ is an $\nin$ to $1$-qubit quantum channel that is implemented by a depth-$d$ $\QACZ$ circuit on $n$ input qubits and $\mathrm{poly}(n)$ auxiliary qubits with $\size$ Toffoli gates. Then for all $k \in [\nin+1]$, we have
    \[\sum_{\substack{P\in \paulis[n+1]} : \\|P| >k} \abs{\wh{\choir_{\chan{}}}(P)}^2 \leq \mathrm{poly}(s) \cdot 2^{-\Omega\pbra{k^{1/d}}} \]    
\end{conjecture}

In particular, we expect no dependence on the number of auxiliary qubits in our spectral bound.
Note that \Cref{conj:qac0-spetral-conc} would immediately imply an average-case lower bound for parity as well as a lower bound for majority against $\QACZ$ circuits with polynomially many auxiliary qubits, and extend the guarantees of our learning algorithm to this broader class of circuits. This would also match the classical bound on the Fourier spectrum of $\ACZ$ circuits obtained by~\cite{LMN:93}.

\paragraph{Improved Learning Algorithms?}

Another natural direction is to improve the runtime of our learning algorithm (see ~\Cref{thm:learning}): while we obtain quasipolynomial sample complexity, we do not provide an explicit algorithm for the final Choi representation rounding step. We conjecture that there exists a quasipolynomial time algorithm implementing an exact rounding procedure, which would also achieve a quasipolynomial runtime for the procedure. Recall that the runtime of the \cite{LMN:93} learning algorithm is (under a strong enough cryptographic assumption) known to be optimal~\cite{kharitonov1993cryptographic}.

\paragraph{Connections to State Synthesis Problems?}

Recent work of Rosenthal~\cite{rosenthal2020bounds} relates the problem of computing parity to various state synthesis problems. Could analytic methods as employed in this paper be used to prove state (or unitary) synthesis lower bounds?

\paragraph{Connections to Pseudorandomness?}

Classically, circuit lower bounds have led to unconditional constructions of pseudorandom generators~\cite{nisan1994hardness}. One could ambitiously hope for unconditional constructions of pseudorandom states against classes of shallow quantum circuits via circuit lower bounds.

\paragraph{An Emerging Analogy?} This work adds to an emerging analogy between Fourier analysis in the classical setting of Boolean functions and ``Pauli analysis'' in the quantum setting of unitary operators or more generally quantum channels~\cite{montanaro2008quantum,chen2023testing,rouze2022quantum,bao2023nearly,volberg2023noncommutative}. Given the tremendous success of Fourier analysis in classical complexity theory, we suspect that much remains to be discovered about the Pauli spectrum of quantum operations.

%% file: sections/prelims.tex
\section{Preliminaries}
\label{sec:prelims}

We assume familiarity with elementary quantum computing and quantum information theory, and refer the interested reader to \cite{Nielsen2010,Wilde2017} for more background. We will sometimes use tensor network diagrams as expository aids; we refer the unfamiliar but interested reader to~\cite{zeph-videos,bridgeman2017hand} for an introduction. 

For $n \geq 1$, we will write $N = 2^n$ and $[n] := \{1,\ldots,n\}$. We denote the Bell state (or EPR pair) on $2n$ qubits by 
\[\ket{\bellstate_n} := \sum_{x\in\zo^n}\ket{x}\otimes\ket{x}.\]
Note in particular that $\ket{\bellstate_n}$ is \emph{not} a normalized state.

 We will write $I_N$ or $I^{\otimes n}$ to denote the $N\times N$ identity matrix; when $N$ is clear from context, we will simply write $I$ instead. We will write $\calM_{N}$ to denote the set of linear operators from $\C^N$ to $\C^N$ and denote by $\calU_{N}$ the set of $N$-dimensional unitary operators, i.e. 
\[\calU_{N} := \cbra{ U \in \calM_{N} : UU^\dagger = U^\dagger U = I_N}.\]

\begin{definition} \label{def:partial-trace}
Given a linear operator $M\in\calM_N$ and $S\sse[n]$, we define the operator $\Tr_S(U)$ obtained by \emph{tracing out $S$} to be 
\[\Tr_S(M) = \sum_{k \in \{0, 1\}^S} (I_{\overline{S}} \otimes \bra{k}) M (I_{\overline{S}} \otimes \ket{k})\]
where $\overline{S} := n\setminus S$.
\end{definition}

In the above definition, we write $\ket{k}$ for $k\in \zo^S$ to be the $|S|$ qubit state in the computational basis corresponding to the bit-string $k$. Note that \Cref{def:partial-trace} aligns with the fact that the trace of a matrix $M$ is given by 
\[\Tr(M) = \sum_{k \in \{0, 1\}^n} \bra{k} M  \ket{k}.\]

\subsection{The Pauli Decomposition}
\label{subsec:pauli-decomp}

Our notation and terminology will frequently follow \cite{ODonnell2014}. We will view $\calM_N$ as a complex inner-product space equipped with the Hilbert--Schmidt inner product: For $A,B\in\calM_N$, we have 
\[\abra{A,B} := \Tr(A^\dagger B).\]
Recall that the Hilbert--Schmidt inner product induces the Hilbert--Schmidt or Frobenius norm, which is given by 
\[\|A\|_F^2 := \abra{A,A} = \sum_{i,j = 1}^{N} |A_{i,j}|^2.\]
We will frequently make use of the fact that the Frobenius norm is invariant under multiplication by unitaries, i.e. 
$\|A\|_F = \|UA\|_F$ {for all unitary operators} $U$.

It is a standard fact that the set of $2\times 2$ Pauli operators, given by 
\begin{align*}
    I := \begin{pmatrix}
        1 & 0 \\ 0 & 1
    \end{pmatrix} &&
    X := \begin{pmatrix} 0 & 1 \\ 1 & 0 \end{pmatrix} &&
    Y := \begin{pmatrix}
        0 & -i \\ i & 0
    \end{pmatrix} &&
    Z := \begin{pmatrix}
        1 & 0 \\ 0 & -1
    \end{pmatrix}
\end{align*}
forms an orthonormal basis for $\calM_N$ with respect to the Hilbert--Schmidt inner product. By taking (normalized) $n$-fold tensors of the Pauli operators one obtains an orthonormal basis for $\calM_N$ where $N = 2^n$. In particular, we have that 
\begin{equation} \label{eq:pauli-tensor-def}
    \paulis[n] := \cbra{I, X, Y, Z}^{\otimes n}
    \quad
\end{equation}
forms an orthonormal basis for $\calM_N$ with respect to the Hilbert--Schmidt inner product. It follows that every $M\in\calM_N$ has a unique representation as 
\[A := \sum_{P\in\calP_n}\wh{A}(P)P \qquad\text{where}\qquad \wh{A}(P) := \frac{1}{N}\abra{A,P}.\]
We will call $\wh{A}(P)$ the \emph{Pauli coefficient} of $A$ on $P$ and will refer to the collection of Pauli coefficients $\{\wh{A}(P)\}_P$ as the \emph{Pauli spectrum} of $A$. 

\begin{notation}
    Given a matrix $M\in\calM_N$, we will write 
    \[\|\wh{M}\|_2^2 := \sum_{P\in\calP_n} |\wh{M}(P)|^2.\]
\end{notation}

It is easy to see that \emph{Parseval's formula} holds for the Pauli decomposition: Given $A \in \calM_N$, we have
\begin{align}
    \frac{1}{N}\|A\|_F^2 = \frac{1}{N}\abra{A,A} = \sum_{P\in\calP_n} |\wh{A}(P)|^2 = \|\wh{A}\|_2^2. \label{eq:qu_parseval_plancharel}
\end{align}
In the case that $U \in \calU_N$ is a unitary, then $\frac{1}{N}\|U\|_F^2 = \|\wh{U}\|_2^2 = 1$. More generally, we have \emph{Plancherel's formula}:
\begin{align}
\frac{1}{N} \abra{A,B} = \sum_{P\in\calP_n} \wh{A}(P)^\ast\cdot\wh{B}(P).\label{eq:plancherel}
\end{align}

For a single-qubit Pauli operator $Q \in \{I,X,Y,Z\}$, we write 
\[
Q_i := I^{\otimes (i-1)} \otimes Q \otimes I^{\otimes (n - i)},\]
and more generally for a subset $S\sse[n]$ and a single-qubit Pauli operator $Q\in\{I,X,Y,Z\}$, we define 
\begin{equation} \label{eq:pauli-subset-notation}
    Q_S := \bigotimes_{i=1}^n Q^{\one\cbra{i\in S}}
\end{equation}
with the convention that $Q^0 \equiv I$.

For an $n$-qubit Pauli operator $P = Q_1 \otimes \dots \otimes Q_n \in\calP_n$ as in \Cref{eq:pauli-tensor-def}, we define 
\[\supp(P) := \{i\in[n] : Q_i \neq I\} \qquad
\text{and}\qquad 
|P| := |\supp(P)|.
\]
We will call $|P|$ the \emph{degree} of $P$; note that it is the number of qubits that $P$ acts non-trivially on. 
Finally, we borrow the classical notion of \emph{spectral weight}, a well-studied quantity from Fourier analysis of Boolean functions:

\begin{definition} \label{def:pauli-weight}
For $M\in \calM_N$, we define the \emph{Pauli weight of $M$ at level $k$} as
\[\weight{=k}[M] := \sum_{\substack{P\in \paulis[n]\\ : |P| = k}} |\wh{M}(P)|^2.\]
We similarly define $\weight{\leq k}[M]$, $\weight{> k}[M]$, and the like.
\end{definition}

The Fourier spectrum and spectral weight distribution of a Boolean function are intimately connected to various combinatorial as well as complexity-theoretic aspects thereof (e.g.~edge boundaries, learnability, and noise stability); we refer the interested reader to the monographs \cite{ODonnell2014,garban2014noise} for further background.

Note that $\sum_{P} |\wh{M}(P)|^2$ is not in general equal to $1$, as it is in the Boolean function analogue. For the special case where $U$ is unitary, the sum is indeed equal to $1$.

\subsection{Quantum Channels and the Choi Representation}
\label{subsec:choi-stuff}

In this paper, we use the Pauli decomposition to analyze quantum \emph{channels}. Recall that a \emph{quantum channel} $\calE$ from $\nin$ to $\nout$ qubits is a completely-positive and trace-preserving (CPTP) linear map from $\calM_{\Nin}$ to $\calM_{\Nout}$ --- throughout this paper we maintain the convention that $\nin$ denotes the number of input qubits, $\nout$ is the number of output qubits, $\ntot = \nin+\nout$ and $\Nin = 2^\nin, \Nout = 2^\nout, \Ntot = 2^\ntot$. As clearly delineated in ~\cite{watrous2018theory}, there are various representations of linear maps --- this paper will make use of the \emph{Choi representation}.

\begin{definition}[Choi representation, Choi state] \label{def:choi-rep}
    Given a quantum channel $\calE$ mapping $\nin$ qubits to $\nout$ qubits, we define its \emph{Choi representation}, written $\Phi_{\calE} \in \calM_{\Ntot}$, as follows:
    \[\Phi_{\calE} := {(I^{\otimes \nin}\otimes \calE)(\ket{\bellstate_\nin}\!\!\bra{\bellstate_\nin})},\]
    where $\ket{\bellstate_\nin} = \sum_{x\in \zo^\nin} \ket{x} \otimes \ket{x}$ is the \emph{un-normalized} $\nin$-qubit Bell state.
    Furthermore, the \emph{Choi state} is defined as
    \[\chois_\calE =\frac{\choir_{\calE}}{\tr(\choir_{\calE})} =  \frac{\choir_{\calE}}{\Nin} .\]
    We designate the first $\nin$ qubits as the register ``$\inn$'', and the remaining $\nout$ qubits as ``$\out$''. 
\end{definition}

For a particular input state $\rho$ to the channel $\calE$, we can determine $\calE(\rho)$ from the Choi representation $\choir_\calE$ using the following \nameCref{fact:choirep-to-output} which can be easily verified.
\begin{fact}\label{fact:choirep-to-output}
    For each $\rho \in \calM_\Nin$, $\chan{}(\rho) = \tr_{\inn}\pbra{\choirE \pbra{\rho^\top \otimes I^{\otimes \nin}}}$
\end{fact}

We often use the \emph{normalized Frobenius norm} $\frac{1}{\sqrt{M}} \fnorm{\choirE}$ of a Choi representation $\choirE \in \calM_\Ntot$.  The following \nameCref{fact:fnorm-bounds} and \nameCref{proposition:basis-distance-to-choi} highlight this convenience.

\begin{fact}\label{fact:fnorm-bounds}
    For each channel $\calE$ from $\nin$ to $\nout$ qubits,
        $\frac{1}{ \Nout^2} \leq \frac{1}{\Ntot}\fnorm{\choirE}^2 \leq 1$
\end{fact}
\begin{proof}
    In \cite{van2002renyi}, van Dam and Hayden show that Rényi entropy is weakly subadditive. For Rényi entropy of order 2 this means that for quantum state $\rho$ on two subsystems $A, B$ the following is true
    \begin{align*}
        S_{2}(\rho_A)  - S_0(\rho_B) \leq S_2(\rho_{AB}) \leq S_2(\rho_A) + S_0(\rho_B).
    \end{align*}
    Where $S_{2}(\rho) :=-\log\pbra{\fnorm{\rho}^2}$ is Rényi entropy of order 2 and $S_0$ is \emph{max-entropy}: $S_0(\rho) := \log\mathrm{rank}(\rho)$. Applying a logarithm this immediately gives us
    \begin{align}\label{eq:fnorm_renyi}
        \frac{\fnorm{\rho_A}^2}{\mathrm{rank}(\rho_B)}\leq \fnorm{\rho_{AB}}^2 \leq \fnorm{\rho_A}^2 \cdot \mathrm{rank}(\rho_B) .
    \end{align}
    Now let $\chois = \chois_\calE$, and let $A$ be the ``$\inn$'' register (the first $\nin$ qubits), and let $B$ be the ``$\out$'' register (remaining $\nout$ qubits). Since $\rho_B \in \calM_{\Nout}$ we have that $\mathrm{rank}(\rho_B) \leq \Nout$. 
    Furthermore, since $\calE$ is trace preserving, we have that $\rho_A = \tr_{\out}(\chois) = \frac{1}{\Nin} I_{\Nin}$, and so $\fnorm{\rho_A}^2 = \frac{1}{\Nin}$.  By \Cref{eq:fnorm_renyi} we have that
    $\frac{1}{\Nin \Nout} \leq \fnorm{\chois_\calE}^2 \leq \frac{\Nout}{\Nin} $. 
    Since $\choir_\calE = \Nin \chois_{\calE}$, it follows that
    \begin{align*}
        \frac{1}{ \Nout^2} \leq \frac{1}{\Ntot}\fnorm{\choirE}^2 \leq 1
    \end{align*}
    Applying Parseval's equation (\Cref{eq:qu_parseval_plancharel}) this proves the \nameCref{fact:weight-bounds}.
\end{proof}
Indeed when $\nout \leq \nin$ these bounds on the squared normalized Frobenius norm $\frac{1}{M} \fnorm{\choirE}^2$ are tight --- it can be easily verified that the lower bound in \Cref{fact:weight-bounds} is achieved by $\calE(\rho) = \frac{1}{\Nout} I_\Nout$ and the upper bound is achieved by the channel $\calE(\rho) = \tr_{\{\nout+1,\dots, \nin\}}(\rho)$. 

The normalized Frobenius distance between Choi representations $\frac{1}{\sqrt{M}} \fnorm{\choir_{\calE} - \choir_{\calF}}$ provides a notion of average-case distance between channels in the following sense.
\begin{proposition}\label{proposition:basis-distance-to-choi}
        For channels $\calE, \calF$ mapping $\nin$ to $\nout$ qubits, and for orthonormal Basis $\calB$ of $\CC^{\Nin}$, we have that 
        \[\E_{\ket{\psi} \sim \calB} \fnorm{\calE(\psi) - \calF(\psi)}^2 \leq \Nout \cdot \frac{1}{\Ntot} \fnorm{\choir_\calE - \choir_{\calF}}^2 \]
        with equality when $\calE(\ketbra{\psi}{\psi'}) = \calF(\ketbra{\psi}{\psi'}) = 0$ for all $\ket{\psi} \neq \ket{\psi'} \in \calB$.
    \end{proposition}
    In words, the condition that $\calE(\ketbra{\psi}{\psi'}) = \calF(\ketbra{\psi}{\psi'}) = 0$ for all $\ket{\psi} \neq \ket{\psi'} \in \calB$ can be stated as that the channel $\calE$ is equivalent to the channel where we first measure in basis $\calB$ and then apply $\calE$ (likewise for $\calF$).
    \begin{proof}
        Using the definition of the Frobenius norm, and the fact that $\calB$ forms an orthonormal basis, it is straightforward to verify that
        \begin{align}
            \E_{\ket{\psi} \sim \calB} \fnorm{\calE(\psi) - \calF(\psi)}^2 = \frac{1}{\Nin}\fnorm{\sum_{\ket{\psi} \in \calB} \psi^* \otimes (\calE - \calF)(\psi)}^2. \label{eq:basis-dist:sum-inside}
        \end{align}
        Let $\calN$ be the channel that measures an $\nin$-qubit state in basis $\calB^* = \{\ket{\psi}^* : \ket{\psi}\in \calB \}$. Then
        \begin{align}
            \sum_{\ket{\psi} \in \calB} \psi^* \otimes (\calE - \calF)(\psi) = \sum_{\ket{\psi}, \ket{\psi'} \in \calB} \calN(\ketbra{\psi}{\psi'}^*) \otimes (\calE - \calF)(\ketbra{\psi}{\psi'}) \label{eq:basis-dist:equal-if-meas}
        \end{align}
        Recall that for any orthonormal Basis $\calB$ of $\CC^\Nin$, we can write the $2\nin$ qubit (unnormalized) Bell state as $\ket{\bellstate_\nin} = \sum_{\ket{\psi} \in \calB} \ket{\psi}^* \otimes \ket{\psi}$
        So we can rewrite the Choi representation as
        \begin{align}
            \choir_\calE - \choir_{\calF} = \sum_{\ket{\psi}, \ket{\psi'} \in \calB} \ketbra{\psi}{\psi'}^* \otimes (\calE - \calF)(\ketbra{\psi}{\psi'}) \label{eq:basis-dist:choidiff-in-diff-bell-basis}
        \end{align}
        Combining \Cref{eq:basis-dist:sum-inside,eq:basis-dist:equal-if-meas,eq:basis-dist:choidiff-in-diff-bell-basis} we have that
        \begin{align*}
            \E_{\ket{\psi} \sim \calB} \fnorm{\calE(\psi) - \calF(\psi)}^2 = \frac{1}{N}\fnorm{(\calN \otimes I_\Nout) \pbra{\choir_\calE - \choir_\calF}}^2.
        \end{align*}
        Note that $\calN$ does not increase the Frobenius norm.
        \begin{align*}
            \fnorm{\mathcal{N}(A)}^2 
            = \sum_{\ket{\psi} \in \calB^*} \abs{\bra{\psi} A \ket{\psi}}^2
            \leq \sum_{\ket{\psi}, \ket{\psi'} \in \calB^*} \abs{\bra{\psi} A \ket{\psi'}}^2 = \fnorm{A}^2.
        \end{align*}
        Therefore
        \[\E_{\ket{\psi} \sim \calB} \fnorm{\calE(\psi) - \calF(\psi)}^2  \leq \frac{1}{\Nin} \fnorm{\choir_\calE - \choir_\calF}^2 = \Nout \cdot \frac{1}{\Ntot} \fnorm{\choir_\calE - \choir_\calF}^2,\]
        completing the proof of the inequality. Note that when $\calE(\ketbra{\psi}{\psi'}) = \calF(\ketbra{\psi}{\psi'}) = 0$ for all $\ket{\psi} \neq \ket{\psi'} \in \calB$, applying the channel $\calN$ to the first register of $\choir_\calE$ and $\choir_\calF$ has no affect, so we get an equality.
    \end{proof}

\subsection{Pauli Spectrum of Quantum Channels}
\label{subsec:choi-pauli}
As shown in \Cref{subsec:pauli-decomp}, we can write the Choi representation for channel $\calE$ in terms of its Pauli decomposition.
\begin{align*}
    \choir_\calE = \sum_{P \in \paulis[\ntot]} \wh{\choir}(P) P.
\end{align*}
We refer to the collection of coefficients $\{\wh{\choir_\calE}(P)\}$ as the \emph{Pauli spectrum of $\calE$}. Furthermore, the level (or degree) $k$ Pauli weight of the channel $\calE$ refers to the level $k$ weight of the Choi representation $\weight{=k}[\choirE] = \sum_{P\in \paulis[\ntot]} |\wh{\choirE}(P)|^2$. We say that the channel $\calE$ is \textit{$\epsilon$-concentrated up to degree $k$} if $\weight{>k}[\choirE]  \leq \epsilon$. 

An important dissemblance between Fourier weight of Boolean functions, and Pauli weight of channels, is that for any Boolean function $f$, the total weight is \emph{always} 1: $\sum_k \weight{=k}[f] = 1$. Whereas this is not generally true for quantum channels. As we see by applying Parsevals (\Cref{eq:qu_parseval_plancharel}) to \Cref{fact:fnorm-bounds}, the total weight is still at most $1$ but can be as small as $\frac{1}{L^2}$.
\begin{fact}[Corollary of \Cref{fact:fnorm-bounds}] \label{fact:weight-bounds}
    For each channel $\calE$ from $\nin$ to $\nout$ qubits,
    $\frac{1}{\Nout^2} \leq \sum_{P} |\wh{\choirE}(P)|^2 \leq 1$.
\end{fact}

In \Cref{sec:low-deg-conc}, we will see that channels computed by certain circuit classes have vanishing weight at higher degrees. Throughout the paper we focus on interesting properties of the Pauli weight $\weight{k}[\choirE]$ at level $k$ --- \Cref{fact:weight-bounds} shows that this is related to the fraction of overall weight by a multiplicative factor between $1$ and $\Nout^2$.

\subsection{Single-Output Channels.}
We will primarily be interested in $\nin$ to $1$-qubit quantum channels that are implemented by applying an $(\nin+a)$-qubit unitary $U$ to the input $\nin$-qubit quantum state---together with $a$-qubit auxiliary state $\psi$---and outputting a single ``target'' qubit by tracing out the rest of the qubits as shown in \Cref{fig:circuit-register-labels}. To denote such channels, we introduce the following notation.

\begin{definition}[$\chan{U, \psi}$] \label{def:U-channel-with-ancilla} \label{def:U-aux-choirep}
    Let $U\in\calM_{\Nin A}$ be a unitary acting on $\nin+a$ qubits, and let $\psi\in \calM_A$ be an $a$-qubit quantum state. We define $\chan{U, \psi}$ to be the following  channel mapping $\nin$ qubits to $1$ qubit:
    \[\chan{U, \psi}(\rho) :=  \tr_{{\out}^c}\pbra{U( \rho\otimes \psi) \, U^\dagger}.\]
    Where $\out = \{\nin+a\}$ is the final qubit output by the circuit as in \Cref{fig:circuit-register-labels}. Its corresponding Choi representation is denoted
    $\choir_{U, \psi} := \pbra{I^{\otimes \nin} \otimes \chan{U, \psi}}\pbra{\proj{\bellstate_\nin}}$, with registers labeled as shown in \Cref{fig:single-output-choi}.
    In the setting without auxiliary states, when $a=0$, we denote $\chan{U}(\rho)$ and $\choir_{U}$ for the corresponding channel and its Choi representation. 
\end{definition}

Note that the register labeling, as shown in \Cref{fig:circuit-register-labels,fig:single-output-choi} is different for incoming and outgoing wires.

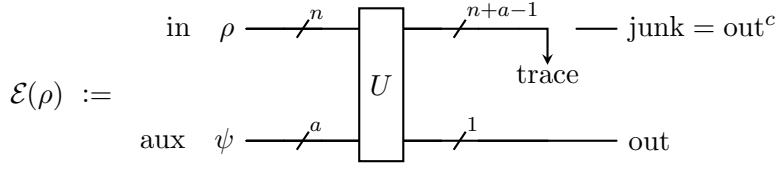
\begin{figure}
    \centering
    \begin{equation*}
    \chan{}(\rho) ~:=~
    \begin{quantikz}
        \lstick{$\inn \quad \rho$} &\qw &\qw \qwbundle{\nin} &\gate[2]{U} &\qw &\qwbundle{\nin+a-1} & \trash{\text{trace}} & \rstick{$\junk= \out^\complement$}\\
        \lstick{$\aux \quad \psi$} &\qw &\qw \qwbundle{a} & \qw &\qw &\qwbundle{1}    & \qw & \rstick{$\out$}\\
    \end{quantikz}
    \end{equation*}
    \caption{The channel $\chan{U, \psi}$ mapping $\nin$ qubits to $1$ qubit by first applying the unitary $U$ to $\nin$ input qubits as well as an $a$-qubit auxiliary state $\ket{\psi}$, then tracing out all but the last qubit --- in the ``$\out$'' register. 
    } 
    \label{fig:circuit-register-labels}
\end{figure}

\input{sections/choi-fig}

Turning to the Choi representation of the auxiliary-free channel $\chan{U}$, the following \nameCref{fact:choirep-altdef} provides us with an alternative definition of $\chan{U}$ that is sometimes more convenient.
\begin{fact} \label{fact:choirep-altdef}
    For each $U\in \unitary{\nin}$, $\choir_U  =  \pbra{ U^\top \otimes I } \pbra{ I^{\otimes (n-1)} \otimes \proj{\bellstate_1} }\pbra{ U^* \otimes I}$
\end{fact}
\begin{proof}
    This proof is illustrated via tensor networks in~\Cref{fig:single-output-choi}.
    \begin{align}
        \choir_U &:= ( I^{\otimes \nin} \otimes \chan{U})\pbra{\proj{\bellstate_\nin}} \nonumber\\
        &= \tr_{\junk}\pbra{\pbra{  I^{\otimes \nin} \otimes U} \proj{\bellstate_\nin} \pbra{ I^{\otimes \nin} \otimes U^\dagger}} \nonumber
    \end{align}
    Using the fact that for any unitary $U\in \unitary{\nin}$, we have
    $(I^{\otimes \nin}\otimes U) \ket{\EPR_\nin} = (U^\top\otimes I^{\otimes \nin})\ket{\EPR_\nin}$, 
    we can further rearrange as
    \begin{align}
        \choir_U &= \tr_{\junk} \pbra{\pbra{ U^\top \otimes I^{\otimes \nin} } \proj{\bellstate_\nin} \pbra{ U^* \otimes I^{\otimes \nin} }}. \nonumber\\
        &= \pbra{ U^\top \otimes I } \tr_{\junk} (\proj{\bellstate_\nin})  \pbra{ U^* \otimes I} \nonumber\\
        &= \pbra{ U^\top \otimes I } \pbra{ I^{\otimes (\nin-1)} \otimes \proj{\bellstate_1} }\pbra{ U^* \otimes I} \label{eq:altdef_choirep}.
\end{align}

\end{proof}

\subsection{\texorpdfstring{$\QACZ$}{QAC0} Circuits}
\label{subsec:qacz}

Our model for quantum circuits with unitary gates is standard and can be found in several textbooks, including~\cite{Nielsen2010}. We also refer the interested reader to a survey by Bera, Green and Homer~\cite{bera2007small} for further background on small-depth quantum circuits.

\begin{definition} \label{def:cz}
    The \emph{controlled phase-flip unitary} on $k$-qubits, written $\CZ_k$, is the unitary operator defined by the following action on the computational basis:
    \[\CZ_k\ket{x_1, \ldots, x_k} = (-1)^{x_1\cdots x_k}\ket{x_1, \ldots, x_k}\]
    for $x_i \in \{0,1\}$. More generally, given a subset $S\sse[n]$ with $|S| = k$, we define $\CZ_S$ to be the unitary operator that acts on the qubits in the registers indexed by $S$ with a $\CZ_{k}$ gate and as the identity on all other qubits, i.e.
    \[\CZ_S\ket{x_1, \ldots, x_n} := I_{}\ket{x_i}_{i\notin S} \otimes \CZ_{k}\ket{x_i}_{i\in S}.\]
\end{definition}

Consider a quantum circuit $C$, written as $C = L_0 M_1 L_1 \ldots M_d L_d$ such that each \emph{layer} $L_i$ consists only of single-qubit gates (i.e. $L_i$ is a tensor product of operators in $\calU_2$) and each $M_k$ is a layer of multi-qubit gates. The \emph{size} of the circuit $C$ is the number of multi-qubit gates in $C$, and the \emph{depth} is the number of layers of multi-qubit gates in $C$.\footnote{In particular, note that single-qubit gates do not count towards circuit depth.}

\begin{definition}[$\QAC$ circuits] \label{def:qac-circuits}
    The class $\QAC$ consists of quantum circuits with $\CZ$ gates and arbitrary single-qubit gates.
\end{definition}

\begin{definition}[$\QACZ$ circuits]
    The class $\QACZ$ contains $\QAC$ circuit families with depth bounded below a constant. That is, the circuit family $\{C_n\}_{n\in \NN}$ is in $\QACZ$ if there exists a $d\in \NN$ such that for each $n\in \NN$, $C_n$ is a $\QAC$ circuit on $n$ qubits with depth at most $d$. 
\end{definition}

We note that the terminology ``$\QACZ$ circuit'' implicitly refers to a $\QACZ$ circuit family; this usage is standard. 

\begin{figure}
    \centering
    \begin{quantikz}[row sep = 2mm]
		& \ghost \vee \qw & \ctrl 1 & \qw & \qw \\
		& \ghost \vee \qw & \ctrl 1 & \qw & \qw \\
		& \ghost \vee \qw & \ctrl 1 & \qw & \qw \\
		\ghost \vee \qw & \qw  & \targ{} & \qw & \qw
    \end{quantikz}   
    \quad = \quad 
    \begin{quantikz}[row sep = 2mm]
		& \ghost \vee \qw & \ctrl 1 & \qw & \qw \\
		& \ghost \vee \qw & \ctrl 1 & \qw & \qw \\
		& \ghost \vee \qw & \ctrl 1 & \qw & \qw \\
		\qw & \gate H  & \gate Z & \gate H & \qw
    \end{quantikz}
    \
    
    \caption{Simulating a Toffoli (i.e. $\mathsf{CNOT}$) gate using $\CZ$ gates and Hadamards.}
    \label{fig:toffoli-cz}
\end{figure}

\begin{remark} \label{remark:toffoli-cz}
    $\QAC$ is sometimes defined using \emph{Toffoli} (i.e.~$\mathsf{CNOT}$) gates instead of $\CZ$ gates. Note that $\mathsf{CNOT}$ gates can be obtained by conjugating the $\CZ$ qubit corresponding to the $\mathsf{CNOT}$ target with Hadamard gates (cf. \Cref{fig:toffoli-cz}). Similarly, $\CZ$ gates can be obtained by conjugating the $\mathsf{CNOT}$ target with Hadamard gates. 
    
    A slight technical advantage of working with $\CZ$ gates over $\mathsf{CNOT}$ gates is the lack of a distinguished ``target'' qubit, which simplifies some of our analyses. More formally, the $\CZ$ gate commutes with the $\mathsf{SWAP}$~operator applied to any two of its qubits.
\end{remark}

Distinct from other works studying $\QACZ$ circuits, this paper is focused on \textit{channels} that are implemented with $\QACZ$ circuits, rather than just unitaries. In this more general framework, we do not assume that the $\QACZ$ circuit ``cleanly computes'' its auxiliary qubits (meaning the auxiliary qubits must start and end as $\ket{0}$. Furthermore, we specifically consider channels implemented with a $\QACZ$ circuit, with only a single designated output qubit.

%% file: sections/choi-fig.tex
\begin{figure}
    \begin{align*}
        {\footnotesize I^{\otimes \nin} \otimes \chan{U}\pbra{\ket{\EPR_\nin}\!\!\bra{\EPR_\nin}}} 
        &=~\diagram{
		\begin{scope}[scale=2]
                \draw[white] (0,0) -- (0,1.5);
                \draw (-0.5,-0.75) .. controls (0,-0.75) and (0,0.75) .. (-0.5,0.75); 
                \draw (-0.75,-0.5) .. controls (-0.25,-0.5) and (-0.25,0.5) .. (-0.75,0.5); 
                \draw (-0.5,0.75) -- (-2,0.75);	
                \draw (-0.75,0.5) -- (-2,0.5); 
                \draw (-0.5,-0.75) -- (-1.5,-0.75); 
                \draw (-0.75,-0.5) -- (-2,-0.5); 
                \draw (0.5,-0.75) .. controls (0,-0.75) and (0,0.75) .. (0.5,0.75); 
                \draw (0.75,-0.5) .. controls (0.25,-0.5) and (0.25,0.5) .. (0.75,0.5); 
                \draw (0.5,0.75) -- (2,0.75);	
                \draw (0.75,0.5) -- (2,0.5); 
                \draw (0.5,-0.75) -- (1.5,-0.75); 
                \draw (0.75,-0.5) -- (2,-0.5); 
                \draw (0.5,-0.75) .. controls (0,-0.75) and (0,0.75) .. (0.5,0.75);
                \draw (0.5,-0.75) -- (1.5,-0.75);
                \draw (0.5,0.75) -- (2,0.75);	
                \draw (1.5,-0.75) .. controls (1.75,-0.75) and (1.75,-1.25) .. (1.5,-1.25); 
                \draw (-1.5,-0.75) .. controls (-1.75,-0.75) and (-1.75,-1.25) .. (-1.5,-1.25); 
                \draw (1.5,-1.25) -- (-1.5,-1.25); 
                \draw[ten] (-1.5,-0.9) rectangle (-0.9,-0.35);
                \draw[ten] (1.5,-0.9) rectangle (0.9,-0.35);
                \node at (-1.2,-0.625) {$U$};
                \node at (1.2,-0.625) {$\,U^\dagger$};
                \node[right] at (2.15,0.75) {\scalebox{.6}{}};
                \node[right] at (2.15,-0.5) {\scalebox{.6}{$\out$}};
                \draw[-latex,dashed,color=black!90!white] (-0.8,-0.85) to [out=0,in=270] (-0.15,-0.6);
                \draw[-latex,dashed,color=black!90!white] (0.8,-0.85) to [out=-180,in=270] (0.15,-0.6);
                \node[] at (0,-1.5) {\scalebox{.6}{$(\nin-1)$ traced out qubits}};
                \draw [decorate,decoration={brace,amplitude=2pt,mirror},yshift=0pt]
(2.15,0.5) -- (2.15,0.75);
                \node[right] (ahhh) at (2.15, 0.625) {\scalebox{.6}{$\inn$}};
		\end{scope}
	}\\
    &=~\diagram{
		\begin{scope}[scale=2]
                \draw[white] (0,0) -- (0,1.5);
                \draw (-0.5,-0.75) .. controls (0,-0.75) and (0,0.75) .. (-0.5,0.75); 
                \draw (-0.75,-0.5) .. controls (-0.25,-0.5) and (-0.25,0.5) .. (-0.75,0.5); 
                \draw (-0.5,0.75) -- (-2,0.75);	
                \draw (-0.75,0.5) -- (-2,0.5); 
                \draw (-0.5,-0.75) -- (-1.5,-0.75); 
                \draw (-0.75,-0.5) -- (-2,-0.5); 
                \draw (0.5,-0.75) .. controls (0,-0.75) and (0,0.75) .. (0.5,0.75); 
                \draw (0.75,-0.5) .. controls (0.25,-0.5) and (0.25,0.5) .. (0.75,0.5); 
                \draw (0.5,0.75) -- (2,0.75);	
                \draw (0.75,0.5) -- (2,0.5); 
                \draw (0.5,-0.75) -- (1.5,-0.75); 
                \draw (0.75,-0.5) -- (2,-0.5); 
                \draw (0.5,-0.75) .. controls (0,-0.75) and (0,0.75) .. (0.5,0.75);
                \draw (0.5,-0.75) -- (1.5,-0.75);
                \draw (0.5,0.75) -- (2,0.75);	
                \draw (1.5,-0.75) .. controls (1.75,-0.75) and (1.75,-1.25) .. (1.5,-1.25); 
                \draw (-1.5,-0.75) .. controls (-1.75,-0.75) and (-1.75,-1.25) .. (-1.5,-1.25); 
                \draw (1.5,-1.25) -- (-1.5,-1.25); 
                \draw[ten] (-1.5,0.35) rectangle (-0.9,0.9);
                \draw[ten] (1.5,0.35) rectangle (0.9,0.9);
                \node[rotate=180] at (-1.2,0.625) {$U$};
                \node[rotate=180] at (1.2,0.625) {$\,U^\dagger$};
                \node[right] at (2.15,0.75) {\scalebox{.6}{}};
                \node[right] at (2.15,-0.5) {\scalebox{.6}{$\out$}};
                \draw [decorate,decoration={brace,amplitude=2pt,mirror},yshift=0pt]
(2.15,0.5) -- (2.15,0.75);
                \node[right] (ahhh) at (2.15, 0.625) {\scalebox{.6}{$\inn$}};
		\end{scope}
	}\\
    &=~\diagram{
		\begin{scope}[scale=2]
                \draw[white] (0,0) -- (0,1.5);
                \draw (-0.75,-0.5) .. controls (-0.25,-0.5) and (-0.25,0.5) .. (-0.75,0.5); 
                \draw (-2,0.75) -- (2,0.75);	
                \draw (-0.75,0.5) -- (-2,0.5); 
                \draw (-0.75,-0.5) -- (-2,-0.5); 
                \draw (0.75,-0.5) .. controls (0.25,-0.5) and (0.25,0.5) .. (0.75,0.5); 
                \draw (0.75,0.5) -- (2,0.5); 
                \draw (0.75,-0.5) -- (2,-0.5); 
                \draw[ten] (-1.5,0.35) rectangle (-0.9,0.9);
                \draw[ten] (1.5,0.35) rectangle (0.9,0.9);
                \node[rotate=180] at (-1.2,0.625) {$U$};
                \node[rotate=180] at (1.2,0.625) {$\,U^\dagger$};
                \node[right] at (2.15,0.75) {\scalebox{.6}{}};
                \node[right] at (2.15,-0.5) {\scalebox{.6}{$\out$}};
                \draw [decorate,decoration={brace,amplitude=2pt,mirror},yshift=0pt]
(2.15,0.5) -- (2.15,0.75);
                \node[right] (ahhh) at (2.15, 0.625) {\scalebox{.6}{$\inn$}};
		\end{scope}
	}
    \end{align*}
    \caption{Tensor network diagram proving alternative definition for $\Phi_{U}$ as stated in \Cref{fact:choirep-altdef}. We adopt the standard convention that \rotatebox[origin=c]{180}{$U$}$\,=U^\top$. Our register labelling follows~\Cref{fig:circuit-register-labels}.}
    \label{fig:single-output-choi}
\end{figure}

%% file: sections/low-deg-conc.tex
\section{Low-Degree Concentration of \texorpdfstring{$\QACZ$}{QAC0} Circuits}
\label{sec:low-deg-conc}

In this section, we will establish the following theorem:
\begin{theorem} \label{thm:pauli-concentration}
    Let $C$ be an $(\nin+a)$-qubit $\QAC$ circuit of depth $d$ and size $\size$. Let  $\psi \in \calM_{2^a}$ be an $a$-qubit quantum state. For each $k \in [n+1]$,
    \begin{align}
        \weight{>k}[\choir_{C, \psi}] \leq O\pbra{s^2 2^{-k^{1/d}}}\cdot  2^a \fnorm{\psi}^2 
        .
    \end{align}
\end{theorem}
We note that if $\psi$ is pure, $\fnorm{\psi}^2 = 1$, but that in general $\fnorm{\psi}^2 \in [2^{-a}, 1]$.
We will first prove \Cref{thm:pauli-concentration} for channels implemented by $\QACZ$ circuits \emph{without} auxiliary qubits (i.e. in the setting when $a = 0$), and will then prove the general statement in \Cref{subsec:postselecting}. Our correlation bound against parity (cf.~\Cref{thm:parity-lb}) will be an immediate consequence of \Cref{thm:pauli-concentration} and will be proved in \Cref{subsec:parity-lb}.

As a corollary (proved in \Cref{appendix:lin-comb-channels}, we see that the result holds also for linear combinations of such channels.
\begin{corollary}
    Consider the channel $\calE(\rho) := \sum_i \alpha_i \calE_i(\rho)$ where $\sum_i \abs{\alpha_i} = 1$ and each $\calE_i = \calE_{U_i, \psi_i}$ for some $(n+a)$-qubit $\QAC$ circuit of depth $d$ and size $\size$, and $a$-qubit state $\psi_i$. Then 
    \[\weight{>k}[\choirE] \leq O\pbra{\size^2 2^{-k^{1/d}}\cdot 2^a \cdot \max_i \fnorm{\psi_i}^2}\]
\end{corollary}

\begin{notation}
    Throughout this section, we will identify a circuit $C$ with the unitary transformation it implements.
\end{notation}

\subsection{Concentration of \texorpdfstring{$\QACZ$}{QAC0} Circuits Without Auxiliary Qubits} 
\label{subsec:conc-no-ancilla}

We will first establish \Cref{thm:pauli-concentration} for channels $\chan{C}$ implemented by a $\QACZ$ $C$ circuit \textit{without} an auxiliary state.

\begin{theorem}[Spectral concentration without auxiliary qubits]\label{thm:anc_free_concentration}
    Suppose $C$ is a depth-$d$, size-$\size$ $\QAC$ circuit on $n$ qubits. 
    Then for each $k \in [n+1]$, we have
    \[\weight{>k}\sbra{\choir_{C}} 
    \leq O\pbra{\size^2 2^{-k^{1/d}}}.\]
\end{theorem}

As discussed in \Cref{subsec:technical-overview}, our argument will proceed in two steps: We will first establish that if a depth-$d$ $\QACZ$ circuit does not have any Toffoli gates of width at least $k^{1/d}$, then its Choi representation has no Pauli weight above level $k$. Next, we will show that deleting such ``wide'' Toffoli gates does not affect the action of the circuit by much.

We start by establishing that circuits with narrow gates have all their Pauli weight on low levels:

\begin{lemma}\label{lem:weight_conc_small_gates}
    Suppose $C$ is a depth-$d$ $\QAC$ circuit such that each gate acts on at most $\ell$ qubits, for some $\ell \geq 2$, and let $\psi$ be an auxiliary state.
    Let $\choir_{C, \psi}$ be as in \Cref{def:U-aux-choirep}.
    Then we have 
    \[\weight{> \ell^d +1}[\choir_{C, \psi}] = 0.\]
\end{lemma}

Note in particular that the bound in \Cref{lem:weight_conc_small_gates} is independent of the number of input qubits~$n$ as well as the number of auxiliary qubits. The proof proceeds via a straightforward lightcone argument (cf. \Cref{fig:lightcone} for lightcone diagram):

\begin{proof}
    Let $C$ be the circuit as described in \Cref{lem:weight_conc_small_gates}. We can view $C$ as a directed graph where each edge is oriented along the direction of computation. Let $L$ be the sub-circuit of $C$ induced by all directed paths connecting the input wires of $C$ with the target qubit (which, by convention, is the last qubit).

    Let $U_{L}$ be the unitary that  $L$ implements and let $V$ be the unitary corresponding to the rest of the circuit. By definition, we have 
    \[C = V U_{L}.\] 
    Since $V$ acts trivially on the target qubit, we have that
    \begin{align*}
        \chan{C, \psi}(\rho) &= \tr_{\inn}\pbra{\pbra{VU_{L}} \,(\rho \otimes \psi)\,\pbra{VU_{L}}^\dagger}
        = \tr_{\inn}\pbra{U_{L} (\rho \otimes \psi)\,U_{L}^\dagger}
        = \chan{U_{L}, \psi}(\rho).
    \end{align*}
    Since $C$ has depth $d$ and gates acting on at most $\ell$ qubits, the size of $L$ is at most $\ell^d$. As $U_{L}$ acts non-trivially on at most $\ell^d$ qubits, $\calE_{U_L, \psi}$, only depends on at most $k = \ell^d$ of the input qubits. Without loss of generality suppose these are the last $k$ qubits, then we can write the channel as $\calE_{C, \psi}(\rho) = \calE_k \pbra{\tr_{\{1, \dots, n-k\}}(\rho)}$ for some $k$ to $1$ -qubit channel $\calE_k$. And 
    \begin{align*}
        \choir_{C, \psi} = I^{\otimes (n-k)}\otimes \choir_{\calE_k}.
    \end{align*}
    Since $\choir_{\calE_k} \in \calM_{2^{k+1}}$, it follows that $\choir_{C, \psi}$ has Pauli degree at most $k +1 = 2^d +1$.
\end{proof}

\Cref{lem:weight_conc_small_gates} tells us that upon removing all the wide gates in a circuit, the resulting channel will have all its weight on low levels. The following lemma shows that in fact (in the auxiliary-free setting) the channel obtained by removing wide gates is close to the original:

\begin{lemma}\label{lem:removing_CZgates}
    Suppose $C \in \calU_N$ is a $\QAC$ circuit with $m$ multi-qubit $\CZ$ gates, each acting on at least $\ell$ qubits. Let $\wt{C} \in \calU_N$ be the circuit obtained by removing these $m$ $\CZ$ gates. Then
    \begin{align}
        \norm{\wh{\choir_{C}} - \wh{\choir_{\wt{C}}}}_2^2\leq  O\left(\frac{m^2}{2^\ell}\right) \label{eqn:lem:removing_CZgates:choirepdist}
    \end{align}
    and
    \begin{align} \label{eqn:weights}
        \weight{> k}\sbra{\choir_{C}} \leq \left(\sqrt{\weight{>k}\sbra{\choir_{\wt{C}}}} + O\left(\sqrt{m^2/2^\ell}\right)\right)^2
    \end{align}
\end{lemma}

Before proving~\Cref{lem:removing_CZgates}, we show how the auxiliary-free concentration bound, \Cref{thm:anc_free_concentration}, follows immediately from \Cref{lem:weight_conc_small_gates,lem:removing_CZgates}.

\begin{proof}[Proof of \Cref{thm:anc_free_concentration}]
    Let $\ell := \ceil{(k-1)^{1/d}}$. Let $\wt{C}$ be the circuit resulting from removing all $m$ of the large $\CZ$ gates (size $\geq \ell$) from $U$. By \Cref{lem:weight_conc_small_gates} we have that
    \[\weight{>k}\sbra{\choir_{\wt{C}}} \leq \weight{>(k-1)^{1/d}}\sbra{\choir_{\wt{C}}} = 0.\] Furthermore, from~\Cref{eqn:weights} (cf. \Cref{lem:removing_CZgates}), we have that 
    \[\weight{>k}\sbra{\choir_{C}} \leq \pbra{\weight{>k}\sbra{\choir_{\wt{C}}} + O\pbra{\sqrt{m^2/2^\ell}}}^2\leq  \frac{32m^2}{2^\ell} \leq O\pbra{\frac{m^2}{2^\ell}},\]
    completing the proof.
\end{proof}

We now turn to the proof of \Cref{lem:removing_CZgates}:

\begin{proof}[Proof of \Cref{lem:removing_CZgates}]
    Let $U, \wt{U} \in \unitary{N}$ be the unitaries corresponding to circuits $C$ and $\wt{C}$ as defined in the statement of \Cref{lem:removing_CZgates}. We will prove the lemma via a sequence of three norm inequalities; we will frequently make use of the unitary invariance of the Frobenius norm.  We start by upper bounding $\fnorm{U - \wt{U}}$.
    \begin{claim}\label{claim:rem_cz_unitary_distance}
        We have $\frac{1}{2^n}\norm{U - \wt{U}}_F^2 \leq 4m^2/2^\ell$.
    \end{claim}
    \begin{proof}
        Let $S_1, \dots, S_m\subseteq [n]$ denote the subsets of qubits that each of the $m$ multi-qubit \CZ gates act on. We can decompose $U$ as 
        \[
        U = U_0\, \CZ_{S_1}\, U_1\, \CZ_{S_2} \dots \CZ_{S_m}\, U_{m} \qquad\text{and}\qquad \wt{U} = U_0\, U_1 \dots U_{m},
        \]
        where $U_0, \dots, U_{m}$ are  $n$-qubit unitaries.
        Let $\wt{U}^{(a)}$ be the unitary resulting from removing the first $a$ $\CZ$ gates from $U$, i.e. 
        \[\wt{U}^{(a)} = U_0\, U_1\, \dots U_{a}\, \CZ_{S_{a+1}}\, U_{a+1}\, \dots \CZ_{S_m}\, U_m.\] 
        Note that $\wt{U}^{(0)} = U$ and $\wt{U}^{(m)} = \wt{U}$. By the triangle inequality, we have
        \begin{equation} \label{eq:triangle-ineq-app1}
            \norm{U- \wt{U}}_F = \norm{ \sum_{a = 1}^m \wt{U}^{(a)} - \wt{U}^{(a-1)}}_F
            \leq \sum_{a = 1}^m \norm{ \wt{U}^{(a)} - \wt{U}^{(a-1)}}_F.
        \end{equation}
        Note, however, that due to the unitary invariance of the Frobenius norm, we have
        \[\fnorm{\wt{U}^{(a)} - \wt{U}^{(a-1)}} = \fnorm{I - \CZ_{S_a}}.\]
        Since $\CZ_{S_a} = I - I_{\overline{S}_a}\otimes2\proj{1^{|S_a|}}$, we have that 
        \[\fnorm{\wt{U}^{(a)} - \wt{U}^{(a-1)}} = \norm{2 \cdot I_{\overline{S}_a}\otimes \proj{1^{|S_a|}}_{S_a}}_F = 2 \cdot 2^{(n-|S_a|)/2}.\]
        Combining this with \Cref{eq:triangle-ineq-app1} and recalling that $|S_a|\geq \ell$, we get that
        \begin{align*}
            \fnorm{U - \wt{U}} & \leq \sum_{a=1}^m \fnorm{\wt{U}^{(a)} - \wt{U}^{(a-1)}}\leq 2m2^{(n-\ell)/ 2},
        \end{align*}
        from which the bound follows. 
    \end{proof}

    We next obtain an upper bound on $\|\choir_U - \choir_{\wt{U}}\|_F$ in terms of $\|U - \wt{U}\|_F$.
    \begin{claim}\label{claim:choi_distance_from_unitary_distance}
        We have 
        \[\norm{\choir_{U} - \choir_{\wt{U}}}_F \leq {4\sqrt{2}}\norm{U - \wt{U}}_F.\]
    \end{claim}
    \begin{proof}
        Recall from \Cref{eq:altdef_choirep} that we can express $\choir_U$ and $\choir_{\wt{U}}$ as
        \begin{align*}
            \choir_{U} = (I \otimes U^\top) \rho_0 (I \otimes U^*) && \choir_{\wt{U}} = (I \otimes \wt{U}^\top) \rho_0 (I \otimes \wt{U}^*)
        \end{align*}
        for $\rho_0 := \tr_{\junk}\left(\proj{\bellstate_n}\right)$
        We can thus write
    \begin{align*}
        \norm{\choir_{U} - \choir_{\wt{U}}}_F &= \norm{(I \otimes U^\top) \rho_0 (I \otimes U^*) - (I \otimes \wt{U}^\top) \rho_0 (I \otimes \wt{U}^*)}_F\\
        &= \left\lVert (I \otimes U^\top) \rho_0 (I \otimes (U^* - \wt{U}^*)) + 
        (I \otimes (U^\top - \wt{U}^\top)) \rho_0 (I \otimes \wt{U}^*) \right\rVert_F\\
        &\leq \norm{\rho_0 \left(I \otimes (U^* - \wt{U}^*)\right)}_F + \norm{\rho_0 \left(I \otimes (U^\top - \wt{U}^\top)\right)}_F\\
        \intertext{thanks to the triangle inequality, and so}
        & = 2\cdot \norm{\rho_0\left(I \otimes (U - \wt{U} )\right)}_F.
    \end{align*} 
    
    Note that $\rho_0 = \tr_{\inn}\left(\proj{\bellstate_n} \right)$ can be written---up to a permutation of qubits---as 
    \[I^{\otimes n-1}\otimes \proj{\bellstate_1}.\] Furthermore, note that $\frac{I^{\otimes n-1}}{2} \otimes \proj{\bellstate_1}$ is a projector,\footnote{The $1/2$ is necessary because $\ket{\bellstate_1} = \ket{00} + \ket{11}$ is un-normalized.} and since projectors do not increase Frobenius norm,\footnote{To see that projectors do not increase Frobenius norm, write $P = \sum_v \lambda_v \proj{v}$ in orthonormal basis $\{\ket{v}\}$ with eigenvalues $\lambda_v \in \{0,1\}$; then $\norm{AP}_F^2 = \sum_v \bra{v} P^\dagger A^\dagger A P \ket{v} = \sum_v \lambda_v \bra{v} A^\dagger A \ket{v} \leq \tr(A^\dagger A)$.} it follows that 
    \begin{align*}
        \norm{\choir_{U} - \choir_{\wt{U}}}_F &\leq 4  \norm{I \otimes (U - \wt{U} )}_F\\
        &\leq 4 \sqrt{2}  \norm{U - \wt{U}}_F
    \end{align*}
    where the last inequality used the sub-multiplicativity of the Frobenius norm and the fact that $\norm{I}_F^2 = 2$.
    \end{proof}

    Note that the first item in \Cref{lem:removing_CZgates} (i.e. \Cref{eqn:lem:removing_CZgates:choirepdist}) follows immediately from \Cref{claim:rem_cz_unitary_distance,claim:choi_distance_from_unitary_distance}:
    \[\frac{1}{2^n} \fnorm{\choir_U - \choir_{\wt{U}}}^2 \leq 32 \cdot \frac{m^2}{2^\ell}  \leq O\pbra{\frac{m^2}{2^{\ell}}}.\]
    
    To prove the second item in \Cref{lem:removing_CZgates} (i.e. \Cref{eqn:weights}), we will upper bound  $\weight{>k}\sbra{\choir_{U}}$ by continuing to build on \Cref{claim:choi_distance_from_unitary_distance}.

    \begin{claim}
        Suppose $\frac{1}{2^n}\norm{U - \wt{U}}_F^2 \leq \delta^2$ and $\weight{>k}\sbra{\choir_{\wt{U}}} \leq \epsilon^2$. Then we have 
        \[\weight{>k}\sbra{\choir_{U}} \leq (\epsilon + 4\sqrt{2}\delta )^2.\]
    \end{claim}
    
    \begin{proof} 
        Using the triangle inequality, we have that for all $k$, the following holds:
        \[\weight{>k}\sbra{A + B} \leq \left( \weight{>k}\sbra{A}^{1/2} + \weight{>k}\sbra{B}^{1/2}\right)^2.\]
        Consequently, we can write 
        \begin{align*}
            \weight{>k}\left[\choir_{U}\right] &= \weight{>k}\left[\choir_{\wt{U}} + \left( \choir_{U} - \choir_{\wt{U}}\right)\right]\\
            &\leq \left(\weight{>k}\left[\choir_{\wt{U}}\right]^{1/2} + \weight{>k}\left[\choir_{U} - \choir_{\wt{U}}\right]^{1/2} \right)^2\\
            &\leq \left(\weight{>k}\left[\choir_{\wt{U}}\right]^{1/2} + \frac{1}{\sqrt{2^{n+1}}}\norm{\choir_{U} - \choir_{\wt{U}}}_F \right)^2\\
            &\leq \left(\epsilon + 4\delta\right)^2
        \end{align*}
        where the second inequality follows from Parseval's formula (\Cref{eq:qu_parseval_plancharel}), and the final inequality relies on \Cref{claim:choi_distance_from_unitary_distance}.
    \end{proof}

This completes the proof of \Cref{lem:removing_CZgates}.
\end{proof}

\subsection{Postselecting for Auxiliary States}
\label{subsec:postselecting}

In this section, we investigate how the spectral concentration of channels is affected when we allow them to make use of an \textit{auxiliary state} (independent of the input). That is, channels of the form
\begin{align}
    \chan{}'(\rho) = \chan{}(\rho \otimes \psi) \qquad\text{for~} \rho \in \calM_N.
\end{align}
where $\chan{}$ is an $(n+a)$ to $1$-qubit channel and $\psi \in \calM_{2^a}$ is an \emph{auxiliary} state. It is helpful to consider auxiliary states as \emph{restricting} the input of the underlying circuit. From this perspective, a $\QACZ$ channel using an auxiliary state is equivalent to an auxiliary-free $\QACZ$ channel with restricted input.

In this section, we will first walk through a straightforward analysis of how the spectral concentration of the Choi representation is affected by either a clean or a dirty auxiliary state. We will then formalize these statements to general auxiliary states.

\subsubsection{``{Clean}'' Auxiliary States}

Suppose $U$ is an $(n+a)$-qubit unitary. We use the shorthand $\chan{U, 0^a}$ to denote $\chan{U, \proj{0^a}}$, or in other words
\[
    \chan{U, 0^a}(\rho) := \chan{U}(\rho \otimes \proj{0^a}).
    \]
We emphasize that now $\chan{U}$ is a channel mapping $n+a$ qubits to $1$ qubit, and so its Choi representation $\choir_{\chan{U}}$ is an $(n+a +1)$-qubit (un-normalized) state. We label the input registers of $\choir_{\chan{U}}$ as ``$\inn$'' ($n$ qubits), and ``$\aux$'' ($a$ qubits), and the output registers as $\mathrm{junk}$ ($n+a-1$ qubits) and ``$\out$'' (1 qubit), i.e. 
\begin{align*}
    \choir_{\chan{U}} = \pbra{(I^{\otimes n}_{\inn} \otimes I^{\otimes a}_{\aux}) \otimes (\chan{U})_{\out}} \pbra{\proj{\bellstate_{n+a}}_{(\inn, \aux), \out}}
\end{align*} 
Conveniently, the Choi representation of $\chan{U, 0^a}$ can be expressed as the result of \emph{postselecting} the auxiliary register $\aux$ of the Choi representation of $\chan{U}$. 
\begin{align}
    \choir_{\chan{U, 0^a}} = \bra{0^a}_{\aux} \choir_{\chan{U}}  \ket{0^a}_{\aux} \label{eq:choi_anc_contract}
\end{align}
where we use $\ket{0^a}_{\aux}$ as shorthand for $(I^{\otimes (n+1)}_{\inn\, \cup\, \out} \otimes \ket{0^a}_{\aux})$. This is easy to see from the following few lines, or the tensor network diagrams in \Cref{fig:choi-with-ancilla}.
\begin{align*}
    \bra{0^a}_{\aux} \choir_{\chan{U}}  \ket{0^a}_{\aux} 
    &= \bra{0^a}_{\aux} \pbra{(I^{\otimes n} \otimes \chan{U})\pbra{\proj{\bellstate_{n+a}}}} \ket{0^a}_{\aux}\\
    &= \bra{0^a}_{\aux} \pbra{\sum_{\substack{x, x' \in \zo^n \\ z,z' \in \zo^{a}}} \ketbra{x}{x'}_{\inn}\otimes \ketbra{z}{z'}_{\aux}\otimes \chan{U}\pbra{\ketbra{x}{x'}\otimes \ketbra{z}{z'}}_{\out}} \ket{0^a}_{\aux}\\
    &= \sum_{x, x' \in \zo^n } \ketbra{x}{x'}_{\inn}\otimes \chan{U}\pbra{\ketbra{x}{x'}\otimes \ketbra{0^a}{0^a}}_{\out}\\
    &= \pbra{I^{\otimes n}\otimes \chan{U, 0^a}}\pbra{\proj{\bellstate_n}} 
    = \choir_{\chan{U, 0^a}}.
\end{align*}

\input{sections/choi-with-ancilla}

This relationship between channels with and without clean auxiliary qubits immediately allows us to extend our concentration results via the following Proposition.
\begin{proposition} \label{prop:conc_change_choi_anc}
    Suppose $\chan{}$ is a quantum channel mapping $(n+a)$ to $1$ qubits, and $\chan{}'$ is the following quantum channel mapping $n$ to $1$ qubits.
    \[
        \chan{}'(\rho) = \chan{}(\rho \otimes \proj{0^a}), \qquad \rho \in \calM_N
    \]
    Then $\weight{>k}[\choir_{\chan{}'}] \leq 2^{a} \cdot \weight{>k}[\choir_{\chan{}}]$ for each $k\in \{0, 1, \dots, n+1\}$.
\end{proposition}
\begin{proof}
By \Cref{eq:choi_anc_contract}, we can express the Pauli coefficients of $\choir_{\chan{}'}$ as 
\begin{align}
    \wh{\choir_{\chan{}'}}(P) = \sum_{Q \in \paulis[a]} \wh{\choir_{\chan{}}}(P\otimes Q) \bra{0^a} Q \ket{0^a} = \sum_{Q\in \{I, Z\}^{\otimes a}} \wh{\choir_{\chan{}}}(P\otimes Q)
\end{align}
for each $P\in \paulis[n+1]$. Now, making use of the Cauchy Schwarz inequality, we get our desired bound on $\weight{>k}[\choir_{\chan{}'}]$:
    \begin{align*}
        \weight{>k}[\choir_{\chan{}'} ]
        &= \sum_{\substack{P \in \paulis[n+1] : \\ |P| >k}} \abs{ \sum_{Q \in \{I, Z\}^{\otimes a}} \wh{\choir_{\chan{}}}(P\otimes Q)}^2\\
        &\leq 2^{a} \sum_{\substack{P \in \paulis[n+1] : \\ |P| >k}} \sum_{Q \in \{I, Z\}^{\otimes a}} \abs{\wh{\choir_{\chan{}}}(P\otimes Q)}^2 \\
        &\leq 2^{a} \weight{>k}[\choir_{\chan{}}]
    \end{align*}
    which completes the proof.
\end{proof}

\subsubsection{``Dirty'' Auxiliary States}
In contrast with a \emph{clean} auxiliary state that is always initialized to the $\ket{0^a}$ state, we consider the case where the auxiliary state given to the circuit is \emph{dirty}, that is there is no guarantee what the auxiliary state will be, yet the circuit still has the desired behavior. This corresponds to an auxiliary state that is maximally mixed ($\psi = \frac{1}{2^a} I^{\otimes a}$):
\begin{align*}
    \chan{}'(\rho) = \chan{}\pbra{\rho \otimes \tfrac{1}{2^a} I^{\otimes a}}.
\end{align*}
Following a similar argument to that used for ``clean'' auxiliary states, one can show that
\begin{align}
    \choir_{\chan{}'} = \frac{1}{2^a} \cdot \tr_\aux \pbra{\choir_{\chan{}}} \qquad \text{and} \qquad \weight{>k}[\choir_{\chan{}'}] \leq \weight{>k}[\choirE], \quad \text{for each } k\geq 1. \label{eq:dirty-change}
\end{align}
As opposed to clean auxiliary states, dirty auxiliary states cannot blow up high-degree Pauli weight, and can be considered ``free.'' 
Below we prove analogous statements for general auxiliary states of which both dirty and clean auxiliary states are a special case.

\subsubsection{Arbitrary Auxiliary States}

More generally, we consider channels that have access to arbitrary auxiliary states, that are not necessarily pure states. As we will see in the following \nameCref{prop:choi_anc_collapse_general}, the Choi representation of such channels is also conveniently related to the auxiliary-free channel.

\begin{proposition} \label{prop:choi_anc_collapse_general}
    Let $\chan{}$ be a quantum channel mapping $n+a$ qubit states to single qubit states, and let $\psi \in \calM_{2^a}$ be a quantum state. Let $\chan{}'$ be the $n$ to $1$ qubit channel resulting from fixing the last $a$ input qubits of $\chan{}$ to $\psi$.
    \begin{align*}
        \chan{}'(\rho) = \chan{}(\rho \otimes \psi) \qquad\text{for}~ \rho \in \calM_{N}.
    \end{align*}
    We label the registers of $\choir_{\chan{}}$, as $\inn$ ($n$ qubits), $\out$ (1 qubit), $\aux$ ($a$ qubits) as follows.
    \begin{align*}
        \choir_{\chan{}} = \pbra{I^{\otimes n}_{\inn} \otimes I^{\otimes a}_{\aux} \otimes \chan{\out}} \pbra{\proj{\bellstate_{n+a}}_{(\inn, \aux), \out}}.
    \end{align*} 
    Then $\choir_{\chan{}'} = \tr_{\aux} \pbra{\choir_{\chan{}} (I^{\otimes n}_{\inn} \otimes (\psi^\top)_{\aux} \otimes I_{\out})}$. 
\end{proposition}
\begin{proof}
We have
    \begin{align*}
        \tr_{\aux}\pbra{\choir_{\chan{}} \pbra{I^{\otimes n +1}_{A^\complement} \otimes \psi^\top_{\aux}}} &= \sum_{\substack{x,x'\in \zo^n \\ y,y'\in \zo^a}} \tr_{\aux}\pbra{\ketbra{x}{x'}\otimes \pbra{\ketbra{y}{y'} \psi^\top}_{\aux} \otimes \chan{}\pbra{\ketbra{x}{x'}\otimes \ketbra{y}{y'}}}\\
        &= \sum_{\substack{x,x'\in \zo^n \\ y,y'\in \zo^a}} \Tr\pbra{\ketbra{y}{y'} \psi^\top} \pbra{\ketbra{x}{x'}\otimes  \chan{}\pbra{\ketbra{x}{x'}\otimes \ketbra{y}{y'}}}\\
        &= \sum_{x,x'\in \zo^n} \ketbra{x}{x'}\otimes  \chan{}\pbra{\ketbra{x}{x'}\otimes \pbra{\sum_{y,y'\in \zo^a} \bra{y'} \psi^\top \ket{y} \ketbra{y}{y'}}}\\
        &= \sum_{x,x'\in \zo^n} \ketbra{x}{x'}\otimes  \chan{}\pbra{\ketbra{x}{x'}\otimes \psi}
        = \choir_{\chan{}}
    \end{align*}
    where in the third equality we used the fact that $\chan{}$ is a linear map.
\end{proof}
Next, we generalize \Cref{prop:conc_change_choi_anc} for arbitrary auxiliary quantum states. Upon restricting $a$ of the input qubits of a channel $\chan{}$ to be the state $\psi \in \calM_{2^a}$, we see that the high degree weight of the Choi representation blows up by at most $2^{a} \fnorm{\psi}^2$.

\begin{proposition}\label{prop:conc_change_choi_anc_general}
    Let $\chan{}, \chan{}', \psi$ be as in \Cref{prop:choi_anc_collapse_general}. Then
    \[\weight{>k}[\choir_{\chan{}'}] \leq 2^{a} \cdot \fnorm{\psi}^2 \cdot \weight{>k}[\choir_{\chan{}}] \text{ for each } k \in \{0,1,\dots n+1\}.\]
\end{proposition}
\begin{proof}
    As shown in \Cref{prop:choi_anc_collapse_general}, we have $\choir_{\chan{}'} = \tr_{\aux} \pbra{\choir_{\chan{}} (I^{\otimes n}_{\inn} \otimes (\psi^\top)_{\aux} \otimes I_{\out})}$. Therefore, 
\begin{align}
    \wh{\choir_{\chan{}'}}(P) = \sum_{Q \in \paulis[a]} \wh{\choir_{\chan{}}}(P\otimes Q) \tr(Q \psi^\top) && \text{for each } P \in \paulis[n+1].
\end{align}
And so,
\begin{align*}
    \weight{>k}[\choir_{\chan{}'}] &= \sum_{\substack{P \in \paulis[n+1] \\ |P| > k}} \abs{\wh{\choir_{\chan{}}}(P)}^2
    = \sum_{\substack{P \in \paulis[n+1] \\ |P| > k}} \abs{\sum_{Q \in \paulis[a]} \wh{\choir_{\chan{}}}(P\otimes Q) \cdot \tr(Q \psi^\top)}^2.
\end{align*}
Using the Cauchy Schwarz inequality, we can upper-bound this with
\begin{align}\label{eq:intermediate:weight_choir_anc_general}
    \weight{>k}[\choir_{\chan{}'}] &\leq  \pbra{ \sum_{\substack{P \in \paulis[n+1] \\ |P| > k}} \sum_{Q \in \paulis[a]} \abs{\wh{\choir_{\chan{}}}(P\otimes Q)}^2} \pbra{\sum_{Q' \in \paulis[a]} \abs{ \tr(Q \psi^\top)}^2}
\end{align}
We can upper-bound the expression inside the first set of parenthesis with 
    \[\sum_{\substack{P\in \paulis[n+1] \\ |P| >k}} \sum_{Q \in \paulis[a]} \abs{\wh{\choir_{\chan{}}}(P\otimes Q)}^2 \leq \weight{>k}[\choir_{\chan{}}].\]
Furthermore, note that for each $Q \in \paulis[a]$, we have $\tr(Q\psi^\top) = 2^a \wh{\psi^\top}(Q)$. Combining this with Parseval's equation (\Cref{eq:qu_parseval_plancharel}), the sum in the second set of parenthesis in \Cref{eq:intermediate:weight_choir_anc_general} is equal to
\begin{align*}
    \sum_{Q' \in \paulis[a]} \abs{ \tr(Q \psi^\top)}^2 = 2^{2a}  \sum_{Q' \in \paulis[a]} \abs{\wh{\psi^\top}(Q')}^2 =  2^a \fnorm{\psi}^2.
\end{align*}
Plugging these into \Cref{eq:intermediate:weight_choir_anc_general} we get our desired inequality.
\begin{align*}
    \weight{>k}[\choir_{\chan{}'}] &\leq 2^{a} \fnorm{\psi}^2 \weight{>k}[\choir_{\chan{}}] 
\end{align*}
which completes the proof.
\end{proof}

Combining \Cref{prop:conc_change_choi_anc_general} and \Cref{thm:anc_free_concentration} we immediately get the following \nameCref{thm:low_deg_conc_formal_arb_anc}.
\begin{theorem}\label{thm:low_deg_conc_formal_arb_anc}
    Let $C$ be an $(n+a)$-qubit $\QAC$ circuit of depth $d$ and size $\size$. Let  $\psi \in \calM_{2^a}$ be an $a$-qubit quantum state. For each $k \in [n+1]$,
    \begin{align}
        \weight{>k}[\choir_{C, \psi}]
        \leq O\pbra{\size^2 \cdot 2^{a - k^{1/d}} \fnorm{\psi}^2}
        .
    \end{align}
\end{theorem}
We note that the settings discussed at the start of this subsection, where the auxiliary state is either \emph{clean} ($\proj{0^a}$) or \emph{dirty} ($2^{-a} I ^{\otimes a}$) are special cases of \Cref{thm:low_deg_conc_formal_arb_anc}.

Note that this concentration bound is unaffected by the number of dirty auxiliary qubits. Furthermore, the bound does not change if instead of $\ket{0^c}$ we used an arbitrary $c$-qubit pure state, including states such as $\ket{\textsf{CAT}_c}$ state which are believed to require $\QAC$ depth $\Omega(\log(c))$ to prepare.

\label{subsec:ancilla}

%% file: sections/choi-with-ancilla.tex
\begin{figure}[t]
    \begin{align*}
        {\footnotesize (I^{\otimes n} \otimes \chan{U})\pbra{\ket{\EPR_n}\!\!\bra{\EPR_n} \otimes \proj{0^a}}} 
        &=~\diagram{
		\begin{scope}[scale=2]
                \draw[white] (0,0) -- (0,1.5);
                \draw (-0.75,-0.5) .. controls (-0.25,-0.5) and (-0.25,0.5) .. (-0.75,0.5); 
                \draw (-1,-0.25) .. controls (-0.5,-0.25) and (-0.5,0.25) .. (-1,0.25); 
                \draw (-0.75,0.5) -- (-2,0.5); 
                \draw (-1,0.25) -- (-2,0.25); 
                \draw (-0.5,-0.75) -- (-1.5,-0.75); 
                \draw (-0.75,-0.5) -- (-1.75,-0.5); 
                \draw (-1,-0.25) -- (-2,-0.25); 
                \draw (0.75,-0.5) .. controls (0.25,-0.5) and (0.25,0.5) .. (0.75,0.5); 
                \draw (1,-0.25) .. controls (0.5,-0.25) and (0.5,0.25) .. (1,0.25); 
                \draw (0.75,0.5) -- (2,0.5); 
                \draw (1,0.25) -- (2,0.25); 
                \draw (0.75,-0.5) -- (1.75,-0.5); 
                \draw (1,-0.25) -- (2,-0.25); 
                \draw (0.5,-0.75) -- (1.5,-0.75);
                \draw (1.75,-0.5) .. controls (2,-0.5) and (2,-1.5) .. (1.75,-1.5); 
                \draw (1.5,-0.75) .. controls (1.75,-0.75) and (1.75,-1.25) .. (1.5,-1.25); 
                \draw (-1.5,-0.75) .. controls (-1.75,-0.75) and (-1.75,-1.25) .. (-1.5,-1.25); 
                \draw (-1.75,-0.5) .. controls (-2,-0.5) and (-2,-1.5) .. (-1.75,-1.5); 
                \draw (1.5,-1.25) -- (-1.5,-1.25); 
                \draw (-1.75, -1.5) -- (1.75, -1.5); 
                \draw[ten] (-1.5,-0.85) rectangle (-0.9,-0.1);
                \draw[ten] (1.5,-0.85) rectangle (0.9,-0.1);
                \node at (-1.2,-0.475) {$U$};
                \node at (1.2,-0.475) {$\,U^\dagger$};
                \node[right] at (2.15,0.75) {};
                \node[right] at (2.15,0.375) {\scalebox{.6}{$\inn$}};
                \node[right] at (2.15,-0.25) {\scalebox{.6}{$\out$}};
                \node[right] at (-0.55,-0.75) {\scalebox{.6}{$\ket{0^a}$}};
                \node[left] at (0.55,-0.75) {\scalebox{.6}{$\bra{0^a}$}};
                \draw [decorate,decoration={brace,amplitude=2pt,mirror},yshift=0pt]
(2.15,0.25) -- (2.15,0.5);
		\end{scope}
	}\\
    &=~\diagram{
		\begin{scope}[scale=2]
                \draw[white] (0,0) -- (0,1.5);
                \draw (-0.5,-0.75) .. controls (0,-0.75) and (0,0.75) .. (-0.5,0.75); 
                \draw (-0.75,-0.5) .. controls (-0.25,-0.5) and (-0.25,0.5) .. (-0.75,0.5); 
                \draw (-1,-0.25) .. controls (-0.5,-0.25) and (-0.5,0.25) .. (-1,0.25); 
                \draw (-0.5,0.75) -- (-1.6,0.75);	
                \draw (-0.75,0.5) -- (-2,0.5); 
                \draw (-1,0.25) -- (-2,0.25); 
                \draw (-0.5,-0.75) -- (-1.5,-0.75); 
                \draw (-0.75,-0.5) -- (-1.75,-0.5); 
                \draw (-1,-0.25) -- (-2,-0.25); 
                \draw (0.5,-0.75) .. controls (0,-0.75) and (0,0.75) .. (0.5,0.75); 
                \draw (0.75,-0.5) .. controls (0.25,-0.5) and (0.25,0.5) .. (0.75,0.5); 
                \draw (1,-0.25) .. controls (0.5,-0.25) and (0.5,0.25) .. (1,0.25); 
                \draw (0.5,0.75) -- (1.6,0.75);	
                \draw (0.75,0.5) -- (2,0.5); 
                \draw (1,0.25) -- (2,0.25); 
                \draw (0.5,-0.75) -- (1.5,-0.75); 
                \draw (0.75,-0.5) -- (1.75,-0.5); 
                \draw (1,-0.25) -- (2,-0.25); 
                \draw (0.5,-0.75) .. controls (0,-0.75) and (0,0.75) .. (0.5,0.75);
                \draw (0.5,-0.75) -- (1.5,-0.75);	
                \draw (1.75,-0.5) .. controls (2,-0.5) and (2,-1.5) .. (1.75,-1.5); 
                \draw (1.5,-0.75) .. controls (1.75,-0.75) and (1.75,-1.25) .. (1.5,-1.25); 
                \draw (-1.5,-0.75) .. controls (-1.75,-0.75) and (-1.75,-1.25) .. (-1.5,-1.25); 
                \draw (-1.75,-0.5) .. controls (-2,-0.5) and (-2,-1.5) .. (-1.75,-1.5); 
                \draw (1.5,-1.25) -- (-1.5,-1.25); 
                \draw (-1.75, -1.5) -- (1.75, -1.5); 
                \draw[ten] (-1.5,0.1) rectangle (-0.9,0.85);
                \draw[ten] (1.5,0.1) rectangle (0.9,0.85);
                \node[rotate=180] at (-1.2,0.475) {$U$};
                \node[rotate=180] at (1.2,0.475) {$\,U^\dagger$};
                \node[right] at (1.5,0.75) {\scalebox{.6}{$\ket{0^a}$}};
                \node[left] at (-1.5,0.75) {\scalebox{.6}{$\bra{0^a}$}};
                \node[right] at (2.15,0.375) {\scalebox{.6}{$\inn$}};
                \node[right] at (2.15,-0.25) {\scalebox{.6}{$\out$}};
                \draw [decorate,decoration={brace,amplitude=2pt,mirror},yshift=0pt]
(2.15,0.25) -- (2.15,0.5);                
                \draw [decorate,decoration={brace,amplitude=2pt},yshift=0pt]
(2.15,-0.5) -- (2.15,-1.5); 
                \node [right] (ahhhhhh) at (2.15, -1) {\scalebox{.6}{$\mathrm{junk}$}};
		\end{scope}
	}\\
    &=~\diagram{
		\begin{scope}[scale=2]
                \draw[white] (0,0) -- (0,1.5);
                \draw (-0.5,-0.75) .. controls (0,-0.75) and (0,0.75) .. (-0.5,0.75); 
                \draw (-0.75,-0.5) .. controls (-0.25,-0.5) and (-0.25,0.5) .. (-0.75,0.5); 
                \draw (-1,-0.25) .. controls (-0.5,-0.25) and (-0.5,0.25) .. (-1,0.25); 
                \draw (-0.5,0.75) -- (-1.6,0.75);	
                \draw (-0.75,0.5) -- (-2,0.5); 
                \draw (-1,0.25) -- (-2,0.25); 
                \draw (-0.5,-0.75) -- (-1.5,-0.75); 
                \draw (-0.75,-0.5) -- (-1.75,-0.5); 
                \draw (-1,-0.25) -- (-2,-0.25); 
                \draw (0.5,-0.75) .. controls (0,-0.75) and (0,0.75) .. (0.5,0.75); 
                \draw (0.75,-0.5) .. controls (0.25,-0.5) and (0.25,0.5) .. (0.75,0.5); 
                \draw (1,-0.25) .. controls (0.5,-0.25) and (0.5,0.25) .. (1,0.25); 
                \draw (0.5,0.75) -- (1.6,0.75);	
                \draw (0.75,0.5) -- (2,0.5); 
                \draw (1,0.25) -- (2,0.25); 
                \draw (0.5,-0.75) -- (1.5,-0.75); 
                \draw (0.75,-0.5) -- (1.75,-0.5); 
                \draw (1,-0.25) -- (2,-0.25); 
                \draw (0.5,-0.75) .. controls (0,-0.75) and (0,0.75) .. (0.5,0.75);
                \draw (0.5,-0.75) -- (1.5,-0.75);	
                \draw (1.75,-0.5) .. controls (2,-0.5) and (2,-1.5) .. (1.75,-1.5); 
                \draw (1.5,-0.75) .. controls (1.75,-0.75) and (1.75,-1.25) .. (1.5,-1.25); 
                \draw (-1.5,-0.75) .. controls (-1.75,-0.75) and (-1.75,-1.25) .. (-1.5,-1.25); 
                \draw (-1.75,-0.5) .. controls (-2,-0.5) and (-2,-1.5) .. (-1.75,-1.5); 
                \draw (1.5,-1.25) -- (-1.5,-1.25); 
                \draw (-1.75, -1.5) -- (1.75, -1.5); 
                \draw[ten] (-1.5,-0.1) rectangle (-0.9,-0.85);
                \draw[ten] (1.5,-0.1) rectangle (0.9,-0.85);
                \node[] at (-1.2,-0.475) {$U$};
                \node[] at (1.2,-0.475) {$\,U^\dagger$};
                \node[right] at (1.5,0.75) {\scalebox{.6}{$\ket{0^a}$}};
                \node[left] at (-1.5,0.75) {\scalebox{.6}{$\bra{0^a}$}};
                \node[right] at (2.15,0.375) {\scalebox{.6}{$\inn$}};
                \node[right] at (2.15,-0.25) {\scalebox{.6}{$\out$}};
                \draw [decorate,decoration={brace,amplitude=2pt,mirror},yshift=0pt]
(2.15,0.25) -- (2.15,0.5);                
		\end{scope}
	}\\[1em]
        &= \bra{0^a}_{\aux} \pbra{(I^{\otimes n} \otimes \chan{U})\pbra{\proj{\bellstate_{n+a}}}} \ket{0^a}_{\aux}
    \end{align*}
    
    \caption{Tensor network diagram for $\Phi_{U}$ where we adopt the standard convention that \rotatebox[origin=c]{180}{$U$}$\,=U^\top$ and write $m := n + a$ for the total number of qubits. Our register labelling follows~\Cref{fig:circuit-register-labels}.}
    \label{fig:choi-with-ancilla}
\end{figure}

%% file: sections/parity-lb.tex
\subsection{\texorpdfstring{$\QACZ$}{QAC0} Lower Bounds for Parity and Majority}
\label{subsec:parity-lb}
In this section, we consider Boolean functions $f
\isazofunc$ and show that when viewing them as a quantum channel $\chan{f}$, their Pauli spectrum is essentially the same as their Fourier spectrum. Using this observation, we connect our low-degree concentration results from the previous section to prove correlation bounds against $\QACZ$ circuits with limited auxiliary qubits for computing high-degree Boolean functions. In particular, we provide explicit bounds for $\Parity_n$ and $\Majority_n$ functions.

The following proposition relates the Fourier spectrum of a Boolean function $f\isazofunc$ to the Pauli spectrum of this function, viewed as a quantum channel as discussed above. We refer the reader to~\cite{ODonnell2014} for further background.

\begin{proposition} \label{prop:fourier-to-pauli}
    Suppose $f\isazofunc$ is a Boolean function and $\rho$. Let $\calE_f$ be the channel mapping $n$-qubit input state $\rho$ to a single-qubit state as follows:
    \begin{align} \label{eq:canonical-function-chanel}
      \calE_f(\rho) := \sum_{x\in\zo^n} \bra{x}\rho\ket{x} \ket{f(x)}\!\!\bra{f(x)}, && \choir_{\chan{f}} = \sum_{x\in \zo^n} \proj{x} \otimes \proj{f(x)}.  
    \end{align}
    Then we have 
    \[\Phi_{\calE_f} = \frac{1}{2}\pbra{I^{\otimes (n+1)} + \sum_{S\sse[n]} \wh{f}(S)Z_S\otimes Z}\]
    where $Z_S$ is defined in \Cref{eq:pauli-subset-notation} and $\wh{f}(S) = \E_{x\sim \zo^\nin}[(-1)^{f(x)} \chi_S(x)]$ \footnote{Recall from \Cref{sec:LMN} that $\chi_S(x):= \prod_{i\in S} x_i$}. 
\end{proposition}

\begin{proof}
    For convenience, we will write $\calE := \calE_f$ and $\Phi:= \Phi_{\calE_f}$ throughout the proof. Note that 
    \begin{align}
        I^{\otimes n} \otimes \calE \pbra{\ket{\EPR_n}\!\!\bra{\EPR_n}} &= \sum_{x\in\zo^n}\ket{x}\!\!\bra{x}\otimes\pbra{\sum_{y\in\zo^n}\braket{y | x}\!\!\braket{x | y} \ket{f(y)}\!\!\bra{f(y)} } \nonumber\\
        &= \sum_{x\in\zo^n} \ket{x}\!\!\bra{x} \otimes \ket{f(x)}\!\!\bra{f(x)}. \label{eq:potato1}
    \end{align}
    It is straightforward to check that the Pauli decomposition of $\proj{x}$ is given by 
    \[\proj{x} = \frac{1}{2^n}\sum_{y\in\zo^n}(-1)^{x\cdot y} Z_y\]
    where we identify $y\in\zo^n$ with the indicator string of the set $S := \{i \in [n]: y_i = 1\}$. Combining this with \Cref{eq:potato1}, we get that 
    \begin{align*}
        I^{\otimes n}\otimes \calE\pbra{\proj{\EPR_n}} 
        &= \sum_{x\in\zo^n} \pbra{\frac{1}{2^n}\sum_{y\in\zo^n}(-1)^{x\cdot y} Z_y} \otimes \pbra{\frac{1}{2}\sum_{b\in \zo} (-1)^{f(x)\cdot b} Z_b}\\
        &= \frac{1}{2}\sum_{\substack{y\in\zo^n\\b\in \zo}} {\pbra{\frac{1}{2^n} \sum_{x\in\zo^n} (-1)^{x\cdot y + f(x)\cdot b}}} Z_y \otimes Z_b.
    \end{align*}
    Note, however, that 
    \[{{\frac{1}{2^n} \sum_{y\in\zo^n} (-1)^{x\cdot y + f(x)\cdot b}}}
        = \begin{cases}
            \mathbf{1}\cbra{y = 0^n} & \text{if}~b = 0\\
            \wh{f}(y) & \text{if}~b = 1
        \end{cases}\] where we once again identify $y$ with a subset $S := \{i\in[n]:y_i = 1\}$, and so we have
    \begin{align}
        I^{\otimes n}\otimes \calE\pbra{\proj{\EPR_n}} 
        = \frac{1}{2}\pbra{I^{\otimes (n+1)} + \sum_{S\sse[n]} \wh{f}(S)Z_S\otimes Z}. \label{eq:unnormalized_boolean_choirep}
    \end{align}

    Now that we have shown the correspondence between the Pauli spectrum and the Fourier spectrum of $f$, we next bound the correlation between a quantum channel and a Boolean function entirely in terms of their Pauli and Fourier spectra respectively.
    \begin{proposition}\label{prop:bool-correlation-from-pauli-conc}
        Suppose $\chan{}$ is an $n$ to $1$-qubit channel.
        For each $x\in \zo^n$ let random variable $Q_{\chan{}}(x)\in \zo$ 
        be the outcome of measuring $\chan{}(\proj{x})$ in the computational basis.
        Then for each $f\isazofunc$ and each $k\in \{0, 1, \dots, n\}$,
        \begin{align*}
            \Prx_{x\in \zo^n} \sbra{Q_{\chan{}}(x) = f(x)}\leq \frac{1}{2}+ \frac{1}{2}\sqrt{\weight{>k}[f]} + \sqrt{\weight{>k+1}[\choir_{\chan{}}]}.
        \end{align*}
    \end{proposition}
    \begin{proof}
        Let $\chan{}, f$ be as in \Cref{prop:bool-correlation-from-pauli-conc}, and let $\chan{f}$ be as defined in \Cref{prop:fourier-to-pauli}. We refer to the probability that $Q_\calE$ and $f$ output the same result when given a uniformly random input, $\Pr_{x\in \zo^n}[Q_\calE(x) = f(x)]$, as the probability of success. Let $\wt{\calE}$ be the \emph{classical} version of the channel $\calE$. That is, $\wt{\calE}$ is the channel resulting from applying a computational basis measurement before and after applying $\calE$.
        \begin{align*}
            \wt{\calE}(\rho) = \calM_\nout\pbra{\calE\pbra{\calM_\nin(\rho)}}
        \end{align*}
        Where $\calM_k$ denotes the $k$-qubit computational basis measurement channel: for each $x \neq x' \in \zo^k$, $\calM(\ketbra{x}{x'}) =0$.
        
        For each $x\in \zo^n$, we can write the probability that $\chan{}$ when given $\ket{x}$ as input, correctly outputs $\ket{f(x)}$ as
        \begin{align*}
            \Pr\sbra{Q_{\chan{}}(x) = f(x)} = \tr\pbra{\proj{f(x)} \chan{}\pbra{\proj{x}}} 
            = \tr\pbra{\chan{f}\pbra{\proj{x}}^\dagger \wt{\chan{}}\pbra{\proj{x}}}.
        \end{align*}
        We now show that the overall probability of success can be written in terms of $\abra{\choir_f, \choir_\calE}$.
        \begin{align*}
            \Ex_{x\in \zo^n} \Pr \sbra{Q_{\chan{}}(x) = f(x)} &= \frac{1}{2^n}\sum_{x\in \zo^n} \tr\pbra{\chan{f}\pbra{\proj{x}}^\dagger \wt{\chan{}}\pbra{\proj{x}}}\\
            &= \frac{1}{2^n}\sum_{x, x'\in \zo^n} \tr\pbra{\chan{f}\pbra{\ketbra{x}{x'}}^\dagger \wt{\chan{}}\pbra{\ketbra{x}{x'}}}\\
            &= \frac{1}{2^n} \Tr\pbra{\choir_{f}^\dagger \choir_{\wt{\calE}}} = \frac{1}{2^n} \abra{\choir_f, \choir_{\wt{\calE}}}
        \end{align*}
        In the second equality we used the fact that $\chan{f}$ is a classical channel --- for each $x\neq x'$, $\chan{f}(\ketbra{x}{x'}) = 0$.
        By Plancharel's (\Cref{eq:plancherel}) We have that $\frac{1}{2^{n+1}}\abra{\choir_f, \choir_\calE} = \abra{\wh{\choir_f}, \wh{\choir_\calE}}$ since the Choi representations $\choir_f, \choir_\calE$ have dimension $2^{n+1}$. Therefore, 
        \begin{align}
            \Pr_x\sbra{Q_{\chan{}}(x) = f(x)} = 2\cdot \abra{\wh{\choir_f}, \wh{\choir_{\wt{\calE}}}}\label{eq:prsucc_as_ip}
        \end{align}
        We now have our probability of success in terms of the Pauli coefficients of the Choi representations. 
        Recall from \Cref{prop:fourier-to-pauli} that for each $S\subseteq [n]$, $\wh{\choir_f}(Z_S\otimes Z) = \frac{1}{2}\wh{f}(S)$, $\wh{\choir_f}(I) = \frac{1}{2}$ and all other coefficients are 0.
        Furthermore, since each quantum channel is trace-preserving (by definition), we also have that $\wh{\choir_{\chan{}}}(I) = \wh{\choir_f}(I)= \frac{1}{2}$. Therefore,
        \begin{align}
            \abra{\wh{\choir_f}, \wh{\choir_{\wt{\calE}}}} = \sum_{P\in \paulis[n+1]} \wh{\choir_f}(P)^*\wh{\choir_{\wt{\calE}}}(P) &=  \frac{1}{4} + \frac{1}{2}\sum_{S \subseteq [n]} \wh{f}(S)\wh{\choir_{\wt{\calE}}}(Z_S \otimes Z)
            \label{eq:coeff_ip_1}
        \end{align}
        For any $k\in \{0, 1, \dots, n\}$ we can bound the second term as follows using Cauchy Schwarz.
        \begin{align*}
            \sum_{S \subseteq [n]} \wh{f}(S)\wh{\choir_{\wt{\calE}}}(Z_S \otimes Z)
            &=\sum_{S \subseteq [n] :|S|\leq k} \wh{f}(S)\wh{\choir_{\wt{\calE}}}(Z_S \otimes Z) + \sum_{S \subseteq [n] :|S|> k} \wh{f}(S)\wh{\choir_{\wt{\calE}}}(Z_S \otimes Z)\\
            &\leq \pbra{\sum_{\substack{S\subseteq [n] : \\|S|\leq k} }\wh{f}(S)^2}^{1/2} \pbra{\sum_{\substack{P\in \paulis[\nin+1] : \\ 0 < |P|\leq k+1}}\wh{\choir_{\wt{\calE}}}(P)^2} ^{1/2} + \pbra{\sum_{\substack{S\subseteq [n] : \\|S|> k} } \wh{f}(S)^2}^{1/2} \pbra{\sum_{\substack{P\in \paulis[\nin+1] : \\ |P|> k+1}}\wh{\choir_{\wt{\calE}}}(P)^2} ^{1/2}
        \end{align*}
        Using the fact that $\wt{\calE}$ is classical and Parseval's (\Cref{eq:qu_parseval_plancharel}), 
        \begin{align*}
            \sum_{\substack{P\in \paulis[\nin+1] : \\ 0 < |P|\leq k+1}}\wh{\choir_{\wt{\calE}}}(P)^2 &= \frac{1}{2^{\nin+1}} \fnorm{\choir_{\wt{\calE}}}^2 - \abs{\wh{\choir_{\wt{\calE}}}}^2 \\
            &= \frac{1}{2^{\nin + 1}} \fnorm{\sum_{x\in \zo^n} \proj{x} \otimes \wt{\calE}(\proj{x})}^2 -\frac{1}{4}\\
            &\leq \frac{1}{2} \E_{x\in \zo^n} \fnorm{\wt{\calE}(\proj{x})}^2 - \frac{1}{4} \leq \frac{1}{4}.
        \end{align*}
        Now plugging this back into the previous equation, along with $\sum_{S\subseteq [n]} \wh{f}(S)^2 \leq 1$, we see that
        \begin{align*}
            \sum_{S \subseteq [n]} \wh{f}(S)\wh{\choir_{\wt{\calE}}}(Z_S \otimes Z) &\leq \frac{1}{2}\pbra{\sum_{\substack{S\subseteq [n] : \\|S|\leq k} } \wh{f}(S)^2}^{1/2} + \pbra{\sum_{\substack{P\in \paulis[\nin+1] : \\ |P|> k+1}}\wh{\choir_{\wt{\calE}}}(P)^2} ^{1/2}\\
            &= \frac{1}{2}\sqrt{\weight{\leq k}[f]} + \sqrt{\weight{>k+1}[\choir_{\wt{\calE}}]}
        \end{align*}
        
Combining this with \Cref{eq:prsucc_as_ip,eq:coeff_ip_1} we see that $\Pr_{x\in \zo^n}\sbra{Q_\calE(x) = f(x)} \leq \frac{1}{2} + \frac{1}{2}\sqrt{\weight{\leq k}[f]} + \sqrt{\weight{>k+1}[\choir_{\wt{\calE}}]}$. Finally, we observe that for each $k \in\{0, 1, \dots, n+1\}$, $\weight{k}[\choir_{\wt{\calE}}] \leq \weight{k}[\choir_{\wt{\calE}}]$. This is because $\choir_{\wt{\calE}} = \calM_{\nin + \nout}(\choir_{\calE})$, and the measurement channel preserves all Paulis in $\{I, Z\}^{\otimes n}$, and maps all other Paulis to $0$.
        
    \end{proof}

    Combining \Cref{prop:bool-correlation-from-pauli-conc} with our concentration bound \Cref{thm:pauli-concentration}, directly implies the following theorem:

    \begin{theorem}\label{QAC-corr-bound}
        Suppose $C$ is a $\QAC$ circuit with depth $d$ and size $\size$ and $\psi$ is a $a$-qubit auxiliary state.
        For each $x\in \zo^n$, let $C(x)\in \zo$ be the random variable indicating the standard basis measurement outcome of $\chan{C, \psi}(\proj{x})$. Then for each Boolean function $f\isazofunc$ and $k \in \{0, 1, \dots, n\}$
        \begin{align}
            \Prx_{x\in \zo^n}\sbra{C(x) = f(x)} \leq \frac{1}{2} + \frac{1}{2}\sqrt{\weight{\leq k}[f]} + O\pbra{\size \cdot 2^{-k^{1/d}/2}} \cdot 2^{a/2} \label{eq:QAC-corr-gen}
        \end{align}
        In particular, 
        \begin{align}
            \Prx_x \sbra{C(x) = \Parity(x)} \leq \frac{1}{2} + O\pbra{\size \cdot 2^{-(n^{1/d} -a)/2}} \label{eq:parity-corr}
        \end{align}
        and
        \begin{align}
            \Prx_x \sbra{C(x) = \Majority(x)} \leq 1 - \Omega\pbra{n^{-1/2}} + 2^{-\Omega\pbra{n^{1/d} - a  -\log \size}}.
        \end{align}
    \end{theorem}
    \begin{proof}
        \Cref{eq:QAC-corr-gen} follows directly from combining \Cref{prop:bool-correlation-from-pauli-conc} and \Cref{thm:pauli-concentration}. The Parity bound (\Cref{eq:parity-corr}) is achieved by plugging in $k = n-1$, since $\weight{\leq n-1}[\Parity_n] = 0$. The majority bound is achieved by  $k = \Omega(n)$ since  $\weight{\leq k}\sbra{\Majority} \leq 1- \frac{1}{\sqrt{k}}$ and so $\sqrt{\weight{\leq k}\sbra{\Majority}}\leq1- \frac{1}{2\sqrt{k}}$.
    \end{proof}
    
This completes the proof.
\end{proof}

\begin{remark}
    More generally, we note that \Cref{thm:pauli-concentration} implies that $\QACZ$ circuits with $\frac{1}{2} n^{1/d}$ clean auxiliary qubits cannot compute a Boolean function $f\isazofunc$ that has a non-zero amount of mass on sufficiently high degree Fourier coefficients. As a consequence, as is the case with the results of \cite{LMN:93}, we obtain a lower bound against $\QACZ$ circuits computing Boolean functions with large \emph{total influence} (cf.~Chapter~2 of \cite{ODonnell2014}) and more generally for unitaries with large total influence (in the sense of~\cite{chen2023testing}). 
    
    For the sake of simplicity, however, we restrict our attention to the parity and majority functions and refer the interested reader to Section 5 of \cite{LMN:93} (or alternatively Chapter~4 of \cite{ODonnell2014}) instead.
\end{remark}

%% file: sections/learning.tex
\section{Learning Low-Degree Channels}
\label{subsec:learning}
We will now leverage our concentration result to propose a learning algorithm for $\QACZ$ circuits. More generally, we establish quasipolynomial sample complexity for learning the Choi representation of \textit{any} low-degree concentrated quantum channel. We consider channels $\chan{}$ mapping $\nin$ input qubits to $\nout$ output qubits--- not just single output channels. Recall that the Choi representation of the channel $\Phi_{\chan{}}$ is of dimension $\Ntot\times \Ntot$ where $\Ntot := 2^\ntot$ and $\ntot = \nin + \nout$ (we further denote $\Nin:=2^\nin$ and $\Nout:=2^\nout$). 
Recall that a quantum channel $\calE$ or its Choi representation $\choir_\calE$ is $\epsilon$-\emph{concentrated up to degree} $k$, if $\weight{>k}[\choir_\calE] \leq \epsilon$. 

For our primary discussion of the learning procedure, we will consider a learning model in which we are given access to many copies of Choi state\footnote{Note that Choi state tomography is often referred to as the ``ancilla-assisted" learning model \cite{leung2000,mohensi2008qpt}.}, i.e. the normalized Choi representation 
$\rho_\calE := \frac{\Phi_\calE}{\tr(\Phi_\calE)}= \frac{1}{\Nin}\choirE$.
Indeed, the Choi state is a valid quantum state as it has $\Tr(\rho_\calE)=1$ and is Hermitian. The Choi state can be prepared by applying the channel $\chan{} $ to half of $\nin$ (\textit{normalized}) Bell pairs. Thus, this is weaker than given general query access to $\calE$. We note that a Choi state is a generalization of sampling a random example $(x, f(x))$ of a Boolean function --- as these are equivalent when the channel is $\chan{f}$.

With this background, we can now state the main result of the section:

\begin{theorem}[Quantum Low-Degree Learning] \label{thm:learnlowdeg}
    Let $k\geq 1$ and let $\calE$ be an $\nin$- to $\nout$-qubit quantum channel that is $\epsilon/2$-concentrated up to Pauli degree $k$. Denote $\ntot=\nin+\nout$ as the total number of input and output qubits of the channel. Then there exists an algorithm, which given $O\pbra{(3\ntot)^k\cdot \frac{1}{4^\ell \epsilon}\cdot \log(\ntot^k/\delta)}$ copies of the Choi state,
    outputs a description of quantum channel $\wt{\calE}$ such that $\frac{1}{\Ntot}\fnorm{\choir_{\calE} - \choir_{\wt{\calE}}}^2 \leq \epsilon$ with probability at least $1 - \delta$. 
\end{theorem}

We remark that as stated, the sample complexity improves as the output dimension of the channel grows for a fixed approximation error $\epsilon$. Though, it is not quite true that our algorithm is ``better'' for larger $\ell$. Rather, the notion of approximation used gets weaker as the number of outputs $\ell$ grows. If we instead consider the expected trace distance between the channels output for a random input chosen from orthonormal basis $\calB$ of $\CC^\Nin$, we have $\E_{\ket{\psi} \in \calB} \tnorm{(\calE- \calF)(\psi)} \leq \sqrt{\Nout} \cdot E_{\ket{\psi} \in \calB} \fnorm{(\calE- \calF)(\psi)}$. Therefore, by \Cref{proposition:basis-distance-to-choi} we have that
    \begin{align}
        \E_{\ket{\psi} \in \calB} \tnorm{\calE(\psi) - \calF(\psi)}^2 \leq
        \frac{\Nout^2}{\Ntot}\fnorm{\choir_\calE - \choir_{\calF}}^2 \label{eq:avg-tdist-to-choir-fdist}.
    \end{align}
    So if our goal is to learn the expected trace distance squared up to error $\eta$, then our error for our learning algorithm is $\epsilon = \eta/L^2$  giving sample complexity $\mathrm{polylog}(n^k) \frac{1}{\epsilon} \log(1/\delta)$ independent of $\ell$, as we'd expect. 

In the special case where $\calE$ is the classical channel for Boolean function $f :\zo^\nin \to \zo$, learning approximately low-degree functions is a direct corollary of our result.
\begin{corollary}\label{cor:learning-bool}
    For each $f:\zo^\nin \to \zo$ with Fourier weight that is $2\epsilon$-concentrated up to degree $k$, given $O((3\nin)^k \frac{1}{\epsilon} \log(\nin^k/\delta))$ random examples $(x, f(x))$, we can learn $h: \zo^\nin \to \zo$ such that $\Pr_x[f(x) \neq h(x)] \leq 4\epsilon$
\end{corollary}
\begin{proof}
    First note that $\chois_f = \frac{1}{\Nin}\sum_{x\in \zo^\nin} \proj{x} \otimes \proj{f(x)}$. So, $\chois_f$ is equivalent to being given a random example $(x, f(x))$ for $x$ drawn uniformly from $\zo^\nin$. Furthermore, by \Cref{prop:fourier-to-pauli}, we have that since $f$ is $2\epsilon$-concentrated up to level $k$, then $\calE_f$ is $\epsilon/2$-concentrated up to level $k+1$. Therefore, by \Cref{thm:learnlowdeg} it follows that with $O\pbra{(3\ntot)^k \cdot \frac{1}{4^\ell \epsilon} \cdot \log(\ntot^k/\delta)}$ copies of the Choi state, with probability $1-\delta$, we learn a $\wt{\choir}$ such that $\frac{1}{\Ntot}\fnorm{\choir_f - \wt{\choir}}^2 \leq \epsilon$.

     For the remainder of this proof we relabel the output bits of $f$ from $\zo$ to $\{\pm1\}$ for convenience. Thus the considered Boolean function is $f: \zo^\nin \to \bits$. Adapting \Cref{prop:fourier-to-pauli} with this relabeling, we can write $\choir_f$ as follows.
     \begin{align}
         \choir_f = \sum_{x\in \zo^\nin} \proj{x} \otimes \frac{1}{2} (I + f(x) Z) \ = \ \frac{1}{2} I + \frac{1}{2} \sum_{S\sse [\nin]} \wh{f}(S) Z(S) \otimes Z \label{eq:bool-learning:f}
     \end{align}
    Without loss of generality, we can assume that the Choi representation $\wt{\choir}$ learned by our algorithm has a similar form for some $g: \bits^\nin \to [-1, +1]$ with outputs in the $[-1, +1]$ interval.
    \begin{align}
        \wt{\choir} = \choir_g = \sum_{x\in \zo^\nin} \proj{x} \otimes \frac{1}{2}(I + g(x) Z) \ = \ \frac{1}{2} I + \frac{1}{2} \sum_{S\sse [\nin]} \wh{g}(S) Z(S) \otimes Z \label{eq:bool-learning:g}
    \end{align}
    This is without loss of generality since for each $P \in \paulis[\nin] \setminus \{I, Z\}^{\otimes \nin}$ setting $\wh{g}(P)$ to zero will not increase the error. Note that the corresponding channel $\calE_g$ is now a randomized classical channel: for each $x\in \zo^\nin$,  $\calE_g(\proj{x}) = \frac{1}{2}(1 + g(x)) \proj{0} + \frac{1}{2}(1 - g(x)) \proj{1}$.

    Our algorithm in \Cref{thm:learnlowdeg} learns a description of $\choir_g$ in terms of its Pauli coefficients $\wh{\choir_g}$. From this, we also have the Fourier coefficients of $g$ as $\wh{g}(S) = 2 \cdot \wh{\choir_g}(Z(S) \otimes Z)$. Let $h:\zo^\nin \to \bits$ be the rounded version of $g$, that is $h(x) = \textrm{sign}(g(x))$ for each $x\in \zo^\nin$. (Here $\textrm{sign}(b) = +1$ if $b \geq 0$ and $\textrm{sign}(b) = -1$ if $b < 0$.) We simply output $h$ as our approximation for $f$. All that is left is to bound the probability that $h$ is wrong:
    \begin{align*}
        \Pr_x [ f(x) \neq h(x)] &= \frac{1}{4} \E_x |f(x) - h(x)|^2\\
        &\leq \E_x |f(x) - g(x)|^2\\
        &= \sum_{S\sse [\nin]} |\wh{f}(S) - \wh{g}(S) |^2\\
        &= 4 \cdot \sum_{S\sse [\nin]} |\wh{\choir_f}(Z(S)\otimes Z)  - \wh{\choir_g}(Z(S)\otimes Z)) |^2\\
        &= 4 \cdot \norm{\wh{\choir_f} - \wh{\choir_g}}^2 = \frac{4}{\Ntot} \fnorm{\wh{\choir_f} - \wh{\choir_g}}^2 \leq 4\epsilon.
    \end{align*}
    In the first inequality we used that $|f(x) - h(x)| \leq 2|f(x) - g(x)|$ since $|f(x) - h(x)| = 2$ when $\textrm{sign}(g(x)) \neq f(x)$ and $|f(x) - h(x)| = 0$ otherwise. The remaining equalities follow from Parseval's (\Cref{eq:qu_parseval_plancharel}) and \Cref{eq:bool-learning:f,eq:bool-learning:g}
\end{proof}

We denote $\calF_k$ as the subset of all Paulis of degree at most $k$.
\[\calF_k := \cbra{P\in\calP_\ntot : |P| \leq k}.\]
Note that the size of $|\calF_k|$ is at most $O(\ntot^k)$. Our construction follows a similar ``meta-algorithm'' for learning low-degree Boolean function as outlined in Chapter 3 of \cite{ODonnell2014}, based on original work by \cite{LMN:93}. 
\begin{enumerate}
    \item \textbf{Estimate Pauli Coefficients:} For each of the low degree Paulis $P\in \calF_k$, we first learn an approximation $\wt{\choir}(P)$ for $\wh{\choir}_\calE(P)$ within additive error $\eta$.
    Setting $\eta$ sufficiently small, we get an $\epsilon$-approximation to $\choir_\calE$.
    \item \textbf{Rounding:} Round $\wt{\choir}$ to the nearest $\choir_{\wt{\calE}}$ such that $\wt{\calE}$ is a quantum channel. We design this convex projection to reduce distance from $\choir_\calE$, ensuring that $\choir_{\wt{\calE}}$ is also $\epsilon$-close to $\choir_\calE$.
\end{enumerate}

Before proving \Cref{thm:learnlowdeg}, we first show how to obtain additive $\eta$-approximations of the Choi representation coefficients $\widetilde{\Phi}(P)$ for each low degree Pauli $P \in \mathcal{F}_k$.

\begin{lemma} \label{lem:learn-coeff-choi}
    For each $\delta, \eta \in (0,1)$, and $\nin$- to $\nout$-qubit channel $\chan{}$, there exists an algorithm that given $T = O\pbra{\frac{3^k}{4^\nout \eta^2} \cdot \log(\ntot^k/\delta)}$ copies of $\chois_{\chan{}}$, outputs $\{\wt{\choir_{\calE}}(P) \}_{P\in \calF_k}$ such that $|\wt{\choir_{\calE}}(P) - \wh{\choir}_{\calE}(P) | \leq \eta$ for each $P\in \calF_k$.
\end{lemma}
\begin{proof}
    We denote $\choir = \choir_{\chan{}}$ and $\chois = \chois_{\chan{}}$.
Recall that $\chois = \frac{1}{\Nin} \choir$. So we have that
\[
\tr(P\chois) = \frac{1}{\Nin}\sum_{Q\in \paulis[\ntot]} \wh{\choir}(Q) \tr(PQ) =\frac{1}{\Nin} \wh{\choir}(P) \tr(I^{\otimes \ntot}) = 2^{\nout} \cdot \wh{\choir}(P).
\]
To approximate $\wh{\choir}(Q)$ to within $\eta$ additive error, it is sufficient to approximate observable $\Tr(P\chois)$ to within $\eta_0 = \eta\cdot \Nout$ additive error.
A standard application of the Chernoff bound gives that, with probability $1-\delta_0$, we can obtain a $\pm\eta_0$-estimate $\widetilde{\rho}(Q)$ of $\Tr(Q\rho)$ by performing 
\[O\pbra{\frac{1}{{\eta_0^2}}\cdot\log\pbra{\frac{1}{\delta_0}}}\]
measurements of $\rho$ with respect to Pauli observable $Q$. However, in order to learn all $O(\ntot^k)$ coefficients for $\mathcal{F}_k$ with probability $1-\delta$, by union bound, we must correctly learn each coefficient with probability $1-\delta_0=1-\frac{\delta}{O(\ntot^k)}$. In all, this implies that a sample complexity of 
\[O\pbra{\frac{\ntot^k}{4^\nout \eta^2} \log\left(\frac{\ntot^k}{\delta}\right)}.\]
suffices to learn $\eta$-approximations of all low-degree Choi state coefficients, i.e. $|\wh{\chois}(Q)-\wt{\chois}(Q)| \leq \eta$ for each $Q \in \calF_k$, with probability at least $1-\delta$. Already this gives us an algorithm to learn $\eta$ approximations of the Choi coefficients for $\calF_k$ with sample complexity 

However, we can approximate the set of expectations $\{\tr(P \rho)\}_{P \in \mathcal{F}_k}$ 
with even better sample complexity by leveraging a well-established, sample-optimal quantum tomography procedure, known as Classical Shadow Tomography (CST) due to Huang et al.~\cite{Huang2020} --- which we use as a black-box.

\begin{theorem}[Classical Shadow Tomography \cite{Huang2020}\footnote{Although we reference the original classical shadow tomography work of \cite{Huang2020}, our bounds leverage improved shadow norm results of \cite{huang2022learning} for Pauli observables, as described in Appendix C.} with Low-Degree Pauli Observables]\label{thm:CST} Assume we are given an unknown $\nin$-qubit state $\rho$, and accuracy parameters $\eta_0,\delta \in [0,1]$, 
and  
\begin{align} \label{eqn:CST_sample}
    O \pbra{\frac{3^k}{\eta_0^2}  \log\pbra{\frac{\nin^k}{\delta}} }
\end{align}
copies of $\rho$ that we make random Pauli measurements on. Then with probability at least $(1-\delta)$, we output an estimate $\wt{s}(P) \in [-1, 1]$ for each $P \in \calF_k$ such that $\abs{\wt{s}(P)-\tr(P \rho)} \leq \eta_0$.
\end{theorem}

By applying CST as in \Cref{thm:CST} on the Choi state with $\eta_0 = \eta \cdot \Nout$, we get our desired sample complexity of $O\pbra{\frac{3^k}{4^\nout \eta^2} \log\pbra{\frac{\ntot^k}{\delta}}}$.
\end{proof}
In this work, we do not provide an explicit rounding procedure, but instead treat finding the nearest Choi representation corresponding to a quantum channel as a black box. In \Cref{sec:rounding} we illustrate that the rounding task can be described as a convex optimization problem. As far as the authors are aware, it is currently unknown whether this optimization problem can be solved exactly in $O(|\calF_k|)$ time.

We are now equipped to prove \Cref{thm:learnlowdeg}.

\begin{proof}[Proof of \Cref{thm:learnlowdeg}]
    Let $\eta = \sqrt{\frac{\epsilon}{2|\calF_k|}}$. By 
    \Cref{lem:learn-coeff-choi} we can learn an approximation $\wt{\choir}(P)$ with $|\wt{\choir}(P) - \wh{\choir}(P)| \leq \eta$ for each $P \in \calF_k$ using only $O\pbra{\frac{3^k}{4^\nout \eta^2} \cdot \log(\ntot^k/\delta)}$ copies of Choi state $\chois_{\calE}$. 
    Since $\choir$ is $\epsilon/2$-concentrated up to level $k$ we have that
    \begin{align} \label{eqn:choi_pars}
    \sum_{P\in\pauli_\ntot}\left|\widehat{\Phi}(P)-\widetilde{\Phi}(P)\right|^2&= \sum_{P\in\mathcal{F}_k}\left|\widehat{\Phi}(P)-\widetilde{\Phi}(P)\right|^2 +  \sum_{P\notin\mathcal{F}_k}\left|\widehat{\Phi}(P)\right|^2\\
    &\leq |\calF_k| \cdot \frac{\epsilon}{2|\calF_k|} + \frac{\epsilon}{2} = \epsilon.
    \end{align}
    By Parseval's (\Cref{eq:qu_parseval_plancharel}) it follows that $\frac{1}{\Ntot}\fnorm{\choir - \wt{\choir}}^2 \leq \epsilon$. The next step is to round $\wt{\choir}$ to the nearest Choi representation $\choir_{\wt{\calE}}$ for a quantum channel $\wt{\calE}$. As we will discuss in \Cref{sec:rounding}, this rounding can be formulated as a convex projection which will only further reduce the error. Thus, 
    \begin{align*}
        \frac{1}{\Ntot}\fnorm{\choir - \choir_{\wt{\calE}}}^2 \leq \frac{1}{\Ntot}\fnorm{\choir - \wt{\choir}}^2 \leq \epsilon.
    \end{align*}
\end{proof}

\subsection{Learning \texorpdfstring{$\QACZ$}{QAC0} Channels} \label{sec:qacz_learn}
Combining \Cref{thm:learnlowdeg} with our low-degree concentration bound of $\QACZ$ in \Cref{sec:low-deg-conc}, a learning algorithm for $\QACZ$ circuit immediately follows. All together, the learning algorithm applied to $\QACZ$ channels closely follows the template of the ``low-degree algorithm'' introduced by Linial, Mansour, and Nisan to learn $\ACZ$ circuits (see~\Cref{sec:LMN}). 

\begin{theorem}[Quasipolynomial Learning of $\QACZ$]\label{thm:learnqac0}
   Let $C$ be a $(n+a)$-qubit $\QAC$ circuit of depth-$d$ and size-$\size$ with $a \leq \frac{1}{2}\log\pbra{s^2/\epsilon }$. Let $\calE = \calE_{C, \ket{0^a}}$. For each $\delta>0$ and error $\epsilon$, we can learn a channel $\widetilde{\calE}$, with 
   $\frac{1}{\Ntot}\fnorm{\choir_{\calE} - \choir_{\wt{\calE}}}^2 \leq \epsilon$
   except with probability $\delta$, using $T = n^{O\pbra{\log^d(\size^2/\epsilon)}}\cdot \log\left(\frac{1}{\delta}\right)$ copies of the Choi state.
\end{theorem}
\begin{proof}
    Following from \Cref{thm:low_deg_conc_formal_arb_anc}, the weight of $\chan{}$ above level $k$ is at most
    \begin{align*}
        \weight{>k}[\choir_{\chan{}}] \leq O\pbra{\size^2 2^{-k^{1/d} + a}}.
    \end{align*}
    Setting $\wt{k} = \log^d(\size^2/\epsilon)$ and using that $a\leq \frac{1}{2} \log(\size^2/\epsilon) = \frac{1}{2}\wt{k}^{1/d}$, we see that 
    $\choir_{\chan{}}$ is $O(\epsilon)$-concentrated below $\wt{k}$, i.e. $\weight{>\wt{k}}[\choir_{\chan{}}] \leq O(\epsilon)$.
    
    Noting that the $\QACZ$ channel has $\nout=1$ output qubits and plugging this concentration degree bound into the sample complexity of \Cref{thm:learnlowdeg} implies that the \qalg uses 
    \begin{align*}
        &O \pbra{\frac{ (3\ntot)^{\wt{k}}}{\epsilon\cdot 4}\log\pbra{\frac{\ntot^{\wt{k}}}{\delta}}} = \ntot^{O\pbra{\log^d(s^2/\epsilon)}} \cdot \log(1/\delta). 
    \end{align*}
    copies of the Choi state,
    or pure state queries to learn an $\epsilon$-approximation of $\choir_\calE$ with probability $1-\delta$.
\end{proof}

Note that naive tomography would result in a $O(4^\ntot)$ bound on the sample complexity (and hence runtime) of learning the Choi representation $\Phi$. \Cref{thm:learnqac0} obtains significant savings by requiring only a quasi-polynomial number of samples (When $\epsilon, \log(1/\delta) = O(1/\poly(\ntot))$).

\subsection{Rounding onto a CPTP Map via Convex Optimization}\label{sec:rounding}
\begin{figure}[t]
    \centering
    \includegraphics[width=.5\textwidth]{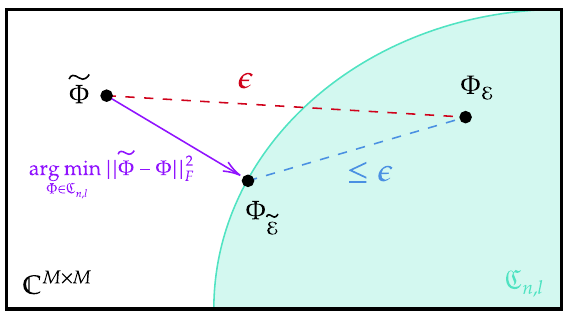}
    \caption{The projection $\Phi_{\widetilde{\calE}} = \underset{\Phi \in \choispace}{\arg\min}\: \|\widetilde{\Phi}-\Phi\|_F^2 \in \choispace$ maps the learned matrix $\widetilde{\Phi}\in\mathbb{C}^{\Ntot\times \Ntot}$ onto the nearest valid Choi representation. Since the space of valid Choi representations, $\choispace$, is a closed convex set, the distance between $\Phi_{\widetilde{\calE}}$ and any other Choi representation must be less than or equal to the distance between $\widetilde{\Phi}$ and any Choi representation. Therefore, if $\widetilde{\Phi}$ is $\epsilon$-close to $\Phi_\calE \in \choispace$, $\Phi_{\widetilde{\calE}}$ is $\epsilon$-close to $\Phi_\calE \in \choispace$.}
    \label{fig:projection}
\end{figure}

The final step of the \qalg outlined in the previous section is to round the learned matrix $\wt{\choir}$, encoding a low-degree approximation of $\choir_\calE$, to a Choi representation $\choir_{\wt{\calE}}$, encoding a completely-positive trace-preserving (CPTP) map $\wt{\calE}$. In order to encode a CPTP map, $\choir_{\wt{\calE}}$ must satisfy\footnote{For a good reference on the mathematical properties of Choi representations and CPTP maps, we refer the interested reader to \cite{watrous2018theory}.} the Complete-Positivity (CP) criterion:
\begin{align} \label{eqn:cp}
    \choir_{\wt{\calE}} \succeq 0 \quad \text{(positive semi-definite)}
\end{align}
and the Trace-Preserving (TP) criterion:
\begin{align} \label{eqn:tp}
    \tr{}_\out(\choir_{\wt{\calE}}) = I_\inn.
\end{align}

Generically, this involves solving a convex constrained minimization problem. Namely, if we let $\choispace$ denote the set of all Choi representation matrices corresponding to an $\nin$ to $\nout$-qubit CPTP, then 
\begin{align} \label{eqn:min_choi}
    \choir_{\wt{\calE}} = \underset{\Phi \in \choispace}{\arg\min}\: \|\widetilde{\Phi}-\Phi\|_F^2.
\end{align}
Since $\choispace$ is a convex set, as demonstrated in \Cref{appendix:convex_proj}, projecting onto $\choispace$ by minimizing with respect to the Frobenius norm reduces the overall error of our learning algorithm. As depicted in \Cref{fig:projection}, with this projection, $\choir_{\wt{\calE}}$ lies closer to the true Choi representation $\Phi_\calE$ than $\wt{\Phi}$ does. Also note that this minimization over the set $\choispace$ can equivalently be expressed as a constrained minimization problem over the larger space of complex matrices $\mathbb{C}^{\Ntot \times \Ntot}$:
\begin{align} \label{eqn:min_e}
    E^* = \underset{E\in\mathbb{C}^{\Ntot \times \Ntot}}{\arg\min} \|E\|_F^2 \hspace{0.1in} \text{ such that } \hspace{0.1in} \begin{cases}
        \tr{}_\out(E) = I_\text{in}-\tr{}_\out(\widetilde{\Phi}) \\
        \widetilde{\Phi} - E \succeq 0
    \end{cases},
\end{align}
where $\choir_{\wt{\calE}}= \widetilde{\Phi} - E^*$. Note, however, that generic convex optimization algorithms generally only guarantee efficient convergence to $\epsilon$-approximate solutions, whereas we need an exact solution to ensure that $\choir_{\wt{\calE}}$ encodes a physical quantum channel.

\subsection{Learning with Measurement Queries}
Finally, in our proof of \Cref{lem:learn-coeff-choi}, we proposed an efficient procedure for approximately learning Pauli coefficients of the Choi representation of quantum channel. Despite the promising sample complexity of this approach, there is a notable physical implementation challenge. Namely, physical implementation would require producing a Choi state, necessitating the preparation and manipulation of a physical Bell state over $2m$-qubits. This, however, is very difficult in practice for large $m$, especially with modern NISQ multi-qubit gate fidelities and state coherence
times.

In light of these challenges, we will now offer a second approach for learning the Choi representation coefficients which is more implementation-friendly. Rather than performing tomography on the Choi state, the physical implementation will amount to preparing $\nin$-qubit input states, applying the channel $\calE$, and performing measurements corresponding to $\nout$-qubit Pauli observables. Notably, all these operations are far more practically feasible than the repeated preparation and measurement of $n$-qubit Choi state $\chois_\calE$. The relevant query model consists of \emph{measurement queries}. Note that, in concurrent work, Wadhwa and Doosti \cite{wadhwa2024learning} formalize this query model as the ``Quantum Process Statistical Query (QPSQ)'' model and study its power for learning arbitrary quantum processes.

\paragraph{Measurement queries} For each measurement query, the learner provides an $\nin$-qubit quantum state $\rho$ along with an observable $O \in \calM_\Nout$, and is given the measurement outcome of $\calE(\rho)$ with respect to observable $O$.

\begin{lemma}
    For each $\delta, \epsilon \in (0,1)$, and $\nin$- to $\nout$-qubit channel $\chan{}$, there exists an algorithm that given $T = O\pbra{\frac{\ntot^k}{4^\ell \eta^2} \cdot \log(1/\delta)}$ measurement queries to $\calE$ outputs $\{\wt{\choir_{\calE}}(P) \}_{P\in \calF_k}$ such that $|\wt{\choir_{\calE}}(P) - \choir_{\calE}(P) | \leq \eta$ for each $P\in \calF_k$.
\end{lemma}
\begin{proof}
We denote $\choir = \choir_{\calE}$. Recall from \Cref{fact:choirep-to-output} that for any input state $\rho \in \calM_{\Nin}$, $\calE(\rho) = \tr_{\inn}(\choir(\rho^\top \otimes I^{\otimes \nin}))$ 
where we label the first $\nin$ qubits of $\choir$ as ``$\inn$'' and the remaining $\nout$ qubits as ``$\out$''. For each Pauli $P\in \paulis[\ntot]$, we let $P_{\inn} \in \paulis[\nin]$ and $P_{\out} \in \paulis[\nout]$ be such that $P = P_{\inn} \otimes P_{\out}$. The key insight of this proof is that measurement of a quantum channel $\calE$, with respect to an input $\rho$ and observable $O$, is related to a linear combination of the Choi representation coefficients $\widehat{\Phi}(P)$ as follows.
\begin{align}
    \tr\left(O \calE(\rho)\right) &= \tr\pbra{O \tr_\inn\pbra{\Phi(\rho_\inn^\top \otimes I_\out )}} \nonumber\\
    &= \sum_{P\in\pauli_\ntot} \wh{\Phi}(P) \tr \pbra{O \tr_\inn\pbra{(P_\inn \otimes P_\out)(\rho_\inn^\top \otimes I_\out)}} \nonumber\\ 
    &=  \sum_{P\in\pauli_\ntot} \widehat{\Phi}(P) \tr (\rho P_\inn) \tr (O P_\out).\nonumber
\end{align}
Given repeated query access to $\calE$, we will see how to learn the Choi representation coefficient $\widehat{\Phi}(P)$ corresponding to any Pauli $P\in\pauli_\ntot$. To this end, there are two related procedures: one for the set of all Paulis $Q\in\{I^{\otimes \nin} \otimes \paulis[\nout]\}$ and another for the remaining Paulis $R\in\pauli_\ntot\backslash\{I^{\otimes \nin} \otimes \paulis[\nout]\}$.

For any $Q\in\{I^{\otimes \nin} \otimes \paulis[\nout]\}$, $\widehat{\Phi}(Q)$ can be estimated via repeated queries of $\calE$ with respect to input $\rho'=\frac{1}{\Nin}I_\inn$ and observable $O'=Q_\out$.\footnote{For the $\QACZ$ channels of interest in this work, the channel output is a single qubit. In this case, $|Q|\leq 1$, meaning $Q$ is guaranteed to lie in the desired set of low-degree Paulis, i.e. $Q\in\mathcal{F}_k$.} Namely, the measurement expectation is proportional to~$\widehat{\Phi}(Q)$.
\begin{align}
    \tr\left(O' \calE(\rho')\right) &= \sum_{P\in\pauli_\ntot} \widehat{\Phi}(P) \tr \left( \frac{1}{\Nin}I_\inn P_\inn\right) \tr (Q_\out P_\out) \nonumber\\
    &=\sum_{P\in\pauli_\ntot} \widehat{\Phi}(P) \cdot \one\cbra{P_\inn=I_\inn} \cdot \Nout\cdot\one\cbra{P_\out=Q_\out} \nonumber\\
    &= \Nout \cdot \widehat{\Phi}(I_\inn\otimes Q_\out) \nonumber\\
    &= \Nout \cdot \widehat{\Phi}(Q).\nonumber
\end{align}
Therefore, to learn $\widehat{\Phi}(Q)$ to error $\eta$, we simply need to learn $\tr\left(O' \calE(\rho')\right)$ to error $\eta_0 = 2^\ell \eta$, 

On the other hand, for each $R\in\pauli_\ntot\backslash\{I^{\otimes \nin} \otimes \paulis[\nout]\}$, to estimate $\widehat{\Phi}(R)$ we consider the expectation of $\calE(\rho^*)$ for input $\rho^*=\frac{1}{\Nin}(I_\inn+R_\inn)$ with respect to observable $O^*=R_\out$. We can write it as follows.
\begin{align*} 
    \tr\left(O^* \calE(\rho^*)\right) &= \sum_{P\in\pauli_\ntot} \widehat{\Phi}(P) \tr \left(\left(\frac{I_\inn+R_\inn}{\Nin}\right) P_\inn\right) \tr (R_\out P_\out) \\
    &= \frac{1}{\Nin} \sum_{P\in\pauli_\ntot} \widehat{\Phi}(P) \left(\tr \left( P_\inn\right)+\tr \left( R_\inn P_\inn\right)\right) \tr (R_\out P_\out) \\
    &= \frac{1}{\Nin} \sum_{P\in\pauli_\ntot} \widehat{\Phi}(P) \left(\Nin\cdot \one\cbra{P_\inn=I_\inn}+\Nin\cdot\one\cbra{P_\inn=R_\inn}\right) \cdot \Nout\cdot\one\cbra{P_\out=R_\out} \\
    &= \Nout \left(\widehat{\Phi}(I_\inn\otimes R_\out)+ \widehat{\Phi}(R_\inn\otimes R_\out)\right) \\
    &= \Nout \left(\widehat{\Phi}(I_\inn\otimes R_\out)+ \widehat{\Phi}(R)\right).
\end{align*}
Note that $\widehat{\Phi}(I_\inn\otimes R_\out)\in \{I^{\otimes \nin} \otimes \paulis[\nout]\}$, so we can estimate $\wh{\choir}(R)$ within error $\eta$ if we have a $2^{\ell-1} \cdot \eta$ estimate of $\tr\pbra{O^* \calE(\rho)}$ and a $\eta/2$ estimate of $\wh{\choir}(R)$. Therefore, to get an additive $\eta$ estimate for $\wh{\choir}(P)$ for each $P\in \calF_k$, it is sufficient for us to estimate $\tr(O' \calE(\rho')$ and $\tr(O^* \calE(\rho^*))$ within error $\eta \cdot 2^\ell /2$ for each $P \in \calF_k$ where $O', O^*, \rho', \rho^*$ depend on $P$ as described above.

Via a Chernoff bound, for each input $\rho$ and observable $P$, $\tr\left(P \calE(\rho)\right)$ can be learned to within additive error $\eta_0$, with probability $1-\delta_0$ via $O(1/\eta_0^2 \cdot \log(1/\delta_0))$ queries. We repeat this process for each $P\in \calF_k$ with $\eta_0 = 2^{\ell-1} \cdot \eta$, and $\delta_0 = \delta/|\calF_k|$. By a union bound, we learn $\eta$-approximations of all $|\mathcal{F}_k|=O(\ntot^k)$ low-degree coefficients with probability at least $1-\delta$ using
\begin{align}
    O\left(\frac{\ntot^k}{4^{\nout}\: \eta^2} \log\left(\frac{\ntot^k}{\delta}\right)\right)
\end{align}
queries to the quantum channel $\calE$. 
\end{proof}

Note that the sample complexity of this procedure is slightly worse than given the Choi state directly, scaling as $O(\ntot^k)$ rather than $O(3^k)$. 
Following the analysis in the Proof of \Cref{thm:learnlowdeg}, this method provides us with a learning algorithm with $O\pbra{\ntot^{2k} \frac{1}{\epsilon} \log(1/\delta)}$ measurement queries.

%% file: sections/combination-channels.tex
\section{Generalizing to linear combinations of channels} \label{appendix:lin-comb-channels}
In this appendix, we show that our result on low-degree Pauli concentration (\Cref{thm:pauli-concentration}) can be generalized from $\QACZ$ channels to linear combinations of such channels. 
\begin{lemma}
    Suppose $\calE$ is a linear combination of channels $\calE_1, \calE_2, \dots, \calE_t$ .
    \begin{align*}
        \calE(\rho) = \sum_{i=1}^t \alpha_i \calE_i(\rho)
    \end{align*}
    such that $\weight{>k}(\choir_{\calE_i}) \leq \epsilon$ for each $i\in [t]$. 
    Then for each $k$,
    \begin{align*}
        \weight{>k}[\choir_\calE] \leq \pbra{\sum_{i=1}^t \abs{\alpha_i}}^2 \epsilon.
    \end{align*}
\end{lemma}
\begin{proof}
     Let $\nin$ and $\nout$ be the number of input and output qubits to $\calE$ respectively, and let $\Nin = 2^\nin, \Nout = 2^\nout$ and $\ntot = \nin + \nout$, $\Ntot = 2^{\ntot}$. For each $i\in [t]$, let $\choir_i = \choir_{\calE_i}$.
    Note that
    \begin{align*}
        \choir_\calE = \pbra{I^{\otimes \nin} \otimes \calE} \pbra{\proj{\bellstate_{\nin}}} = \sum_{i=1}^t (I^{\otimes \nin} \otimes \calE_i) \pbra{\proj{\bellstate_\nin}} = \sum_{i=1}^t \alpha_i \choir_i.
    \end{align*}
    Consider a particular $k$, and let $H_i =\sum_{\substack{P\in \paulis[\ntot] :\\ |P| >k}} \wh{\choir_i}(P) P$. Furthermore, let $\wh{H_i}$ be the vector of Pauli coefficients of $H_i$. So $\weight{>k}(\choir_i) = \norm{\wh{H_i}}_2^2$ for each $i \in [t]$. Now
    \begin{align*}
        \weight{>k}[\choir_{\calE}] &= \norm{\sum_{i = 1}^t \alpha_i \wh{H_i}}_2^2 .
    \end{align*}
    Applying the triangle inequality,
    \begin{align*}
        \weight{>k}[\choir_{\calE}]  &\leq  \pbra{\sum_{i = 1}^t \abs{\alpha_i}\cdot \norm{\wh{H_i}}_2}^2\\
        &= \pbra{\sum_{i=1}^t \abs{\alpha_i} \cdot \sqrt{ \weight{>k}[\choir_i}]}^2\\
        &\leq \pbra{\sum_{i=1}^t \abs{\alpha_i}}^2 \epsilon
    \end{align*}
\end{proof}

\section{Convex CPTP Projection} \label{appendix:convex_proj}
Leveraging the fact that $\widetilde{\Phi}$ is $\epsilon$-close to $\Phi_\calE$, we will now prove that the Choi representation $\Phi_{\widetilde{\calE}}$, resulting from the convex Frobenius norm minimizaton described in \Cref{sec:rounding}, is also $\epsilon$-close to $\Phi_\calE$. Note that, in the case that $\widetilde{\Phi} \in \choispace$, the projection simply maps $\widetilde{\Phi}$ to itself, which implies that $\Phi_{\widetilde{\calE}}=\widetilde{\Phi}$ is $\epsilon$-close to $\Phi_\calE$. Therefore, we only need to consider the case in which $\widetilde{\Phi}_\calE \notin\choispace$.

We begin by decomposing the error between  $\widetilde{\Phi}$ and $\Phi_\calE$ in terms of $\Phi_{\widetilde{\calE}}$, as
\begin{align*} 
    \left\| \widetilde{\Phi}- \Phi_\calE  \right\|_F^2 &= \left\| (\widetilde{\Phi} - \Phi_{\widetilde{\calE}})+(\Phi_{\widetilde{\calE}} -\Phi_\calE)  \right\|_F^2 \\
    &= \left\| \widetilde{\Phi} - \Phi_{\widetilde{\calE}} \right\|_F^2 +2 \left\langle \widetilde{\Phi} - \Phi_{\widetilde{\calE}}, \Phi_{\widetilde{\calE}} -\Phi_\calE\right\rangle +\left\| \Phi_{\widetilde{\calE}} -\Phi_\calE  \right\|_F^2.
\end{align*}
Since the Frobenius norm squared is non-negative, if we can simply prove that 
\begin{align} \label{eqn:inner_prod}
    \left\langle \widetilde{\Phi} - \Phi_{\widetilde{\calE}}, \Phi_{\widetilde{\calE}} -\Phi_\calE\right\rangle \geq 0,
\end{align}
then it follows that
\begin{align*} 
     \left\| \Phi_{\widetilde{\calE}} -\Phi_\calE  \right\|_F^2 \leq \left\| \widetilde{\Phi}- \Phi_\calE  \right\|_F^2 < \epsilon.
\end{align*}
Therefore, to prove that $\Phi_{\widetilde{\calE}}$ is $\epsilon$-close to $\Phi_\calE$, we will prove \Cref{eqn:inner_prod}. 

To do so, note that $\choispace$ is a closed convex set in the Hilbert space $\mathbb{C}^{\Ntot\times \Ntot}$, equipped with Frobenius norm inner product $\langle A, B \rangle = \tr [A^\dagger B]$. Furthermore, in \Cref{eqn:min_choi}, we defined $\Phi_{\widetilde{\calE}}$ to be the minimizer of the Frobenius norm distance with respect to $\widetilde{\Phi}$. Therefore $\Phi_{\widetilde{\calE}}$ must satisfy the First-Order Optimality Condition: 

\begin{fact}[First-Order Optimality Condition]\label{thm:firstopt}%
   For an optimization problem 
   \begin{align*}
       \underset{X\in \mathcal{X}}{\arg\min} f(X)
   \end{align*}
   over closed convex set $\mathcal{X}$ and with convex differentiable $f$, feasible point $X$ is optimal if and only if
   \begin{align*}
       \langle\nabla f(X), Y-X\rangle \geq 0, \:\: \forall Y\in\mathcal{X}.
   \end{align*}
\end{fact}

\noindent From \Cref{eqn:min_choi}, $\Phi_{\widetilde{\calE}}$ is the minimizer of the function $f(\Phi)=\|\Phi -\Phi_\calE \|_F^2$, for which $\nabla f(\Phi)=2 (\Phi-\widetilde{\Phi})$. The First-Order Optimality Condition, therefore, implies that
\begin{align}
    \forall \: \Phi_\calE \in \choispace \: : \: \langle\nabla f(\Phi_{\widetilde{\calE}}), \Phi_\calE -\Phi_{\widetilde{\calE}}\rangle = 2 \langle \Phi_{\widetilde{\calE}}-\widetilde{\Phi}_\calE, \Phi_\calE -\Phi_{\widetilde{\calE}}\rangle \geq 0,
\end{align}
thereby proving \Cref{eqn:inner_prod}.